%% file: branch-misses.tex
\pdfoutput=1
\documentclass[leqno,twocolumn,american,11pt]{article}  
\usepackage{ltexpprt}

\input{preamble}

\makeatletter
\title{Analysis of Branch Misses in Quicksort%
	\thanks{%
		Part of this research was done during a visit at UPC,
		for which the second and third authors acknowledge support by 
		project TIN2007-66523
		\textsl{Formal methods and algorithms for system design (FORMALISM)} 
		of the Spanish Ministry of Economy and Competitiveness
	}
}
\author{%
	Conrado Mart\'\i nez%
	\thanks{%
		Department of Computer Science,
		Univ.\ Polit\`ecnica de Catalunya,
		Email: \texttt{conrado@cs.upc.edu}%
	} 
\and 
	Markus E.\ Nebel%
	\thanks{%
		Computer Science Department, 
		University of Kaiserslautern,
		Email: \texttt{\{wild,nebel\}@cs.uni-kl.de}%
	}
	\thanks{%
		Department of Mathematics and Computer Science, 
		University of Southern Denmark%
	}
\and 
	Sebastian Wild${}^{\@fnsymbol3}$}
\makeatother

\hypersetup{pdftitle={Analysis of Branch Misses in Quicksort},
 pdfauthor={Conrado Mart\'\i nez, Markus E. Nebel and Sebastian Wild}}

\begin{document}

\maketitle

\begin{abstract}
The analysis of algorithms mostly relies on counting classic 
elementary operations like additions, multiplications, comparisons, swaps etc. 
This approach is often sufficient to quantify an algorithm's efficiency. 
In some cases, however, features of modern processor architectures like 
pipelined execution and memory hierarchies have significant impact on running time
and need to be taken into account to get a reliable picture. 
One such example is Quicksort: 
It has been demonstrated experimentally that under certain conditions on the hardware 
the classically optimal balanced choice of the pivot as median of a sample 
gets \emph{harmful}. 
The reason lies in mispredicted branches whose rollback costs become dominating.

In this paper, we give the first precise analytical investigation of 
the influence of pipelining and the resulting branch mispredictions 
on the efficiency of (classic) Quicksort 
and Yaroslavskiy's dual-pivot Quicksort as implemented in Oracle's Java~7 library. 
For the latter it is still not fully understood why experiments prove it
$10\,\%$ faster than a highly engineered implementation of a classic single-pivot version. 
For different branch prediction strategies, we give precise asymptotics 
for the expected number of branch misses caused by the aforementioned Quicksort variants 
when their pivots are chosen from a sample of the input.
We conclude that the difference in branch misses is too small to explain the 
superiority of the dual-pivot algorithm.
\end{abstract}

\todoin[disable]{%
	\begin{itemize}
	\item
		Motivation for our paper (contributions):
		\begin{enumerate}
			\item
				make analyses of \cite{Kaligosi2006branch} and \cite{Biggar2008} more precise
				(they do not report precise leading terms; 
				they either do not consider sampling at all, or only for $k\to\infty$).
			\item
				exclude BM as possible explanation for Yaroslavskiy's superior running time
				(which gives further evidence for the memory hierarchy hypothesis;
				We should avoid the term ``cache miss'' as possibly the bandwidth to 
				main memory is the actual limiting factor, 
				not the cache \emph{misses} themselves
				(buzzword: prefetching of cache lines in linear scans))
			\item
				As shown experimentally, skewed pivots may improve overall performance 
				in cases where BM are very expensive.
				With linear combination cost measure $\mathit{BC} + \xi \mathit{BM}$, we
				can study the overall behavior dependent on $\xi$; 
				find optimal skews as function of $\xi$
			\item
				precise analysis of finite size samples allows to optimize $\vect t$
				for practical $k$
		\end{enumerate}

	\item
		Discuss results:
		\begin{itemize}
		\item Classic Quicksort
			\begin{itemize}
			\item
				threshold for $\xi$ from where on skewed sampling helps
			\item
				optimize $\vect \tau$ as function of $\xi$
			\item
				also optimize $\vect t$ as function of $\xi$ for practical $k$?
			\end{itemize}
		\item Yaroslavskiy
			\begin{itemize}
			\item
				compare BM counts with CQS $\leadsto$
				without sampling: CQS slightly better 
				$\leadsto$ BM not explanation \\
				What happens with sampling?
			\item
				threshold for $\xi$ from where on differently skewed sampling helps
			\item
				optimize $\vect \tau$ as function of $\xi$
			\item
				also optimize $\vect t$ as function of $\xi$ for practical $k$?
			\end{itemize}
		\end{itemize}
		\bigskip
		\item 
			curious fact: cmps and fresh elements coincide for CQS;
			as do swaps and BM (up to constant factor)
	\end{itemize}
	
}

\section{Introduction}

\todoin[disable]{%
	\begin{itemize}
	\item[] \textbf{Introduction}
	\item
		cite previous analysis results on CQS and YQS, 
		including important general results, but focus on BM.
	\item
		Motivation for our paper (contributions):
		\begin{enumerate}
			\item
				As shown experimentally, skewed pivots may improve overall performance 
				in cases where BM are very expensive.
				With linear combination cost measure $\mathit{BC} + \xi \mathit{BM}$, we
				can study the overall behavior dependent on $\xi$; 
				find optimal skews as function of $\xi$
			\item
				precise analysis of finite size samples allows to optimize $\vect t$
				for practical $k$
		\end{enumerate}
	\item[\warning{}] 
		State explicitly what has been done before, 
		esp.\ on YQS and what is new \dots
	\end{itemize}
}

\input{branch-misses-Introduction}

\input{branch-misses-analysis}

\input{branch-misses-discussion}

\small
\bibliography{quicksort-refs}

%\clearpage
\appendix
\section*{Appendix}
\input{branch-misses-notation}
\input{branch-misses-steady-state-predictability}
\input{branch-misses-preliminaries}

\input{branch-misses-computations}

\end{document}

%% file: preamble.tex
\usepackage{lmodern}
\usepackage[utf8]{inputenc}
\usepackage[T1]{fontenc}
\usepackage{babel}
\input{ushyphex.tex} % hyphenation exceptions

\date{\small\today}

\makeatletter
\let\old@subsection\subsection
\def\subsection#1{\old@subsection[#1]{#1.}}
\makeatother

\usepackage[numbers,square]{natbib}
\usepackage{xspace}
\usepackage{enumitem}
\usepackage{array}
\usepackage{booktabs}
\usepackage{multirow}
\usepackage[referable]{threeparttablex}
\usepackage{colortbl}
\usepackage[tight,TABBOTCAP]{subfigure}

\usepackage[textsize=footnotesize,textwidth=2.4cm,
	backgroundcolor=white,
	shadow]{todonotes} 
\presetkeys{todonotes}{fancyline}{} % default: fancyline
\newcommand\todoin[2][]{\todo[inline, caption={2do}, #1]{
	\begin{minipage}{\linewidth-1em}\noindent\relax#2\end{minipage}}}

\usepackage{wref}

\usepackage{xfrac}
\usepackage{amsmath}
\usepackage{amssymb}
\usepackage{dsfont}
\usepackage{bm}
\usepackage{mathtools}
\usepackage{stmaryrd}
\usepackage{tikz}
\usepackage{colonequals}
\usepackage[hyphens]{url}
\usepackage{ifthen}
\usepackage{placeins}
\usepackage{needspace}

\usepackage[margin=.75em,font={small,rm},labelfont={bf}]{caption}
\DeclareCaptionStyle{ruled}{format=plain,%
						    labelfont=bf,%
                            labelsep=period,%
                            strut=on,%
                            margin=.3em,%
                            font={small,rm}%
}

\usetikzlibrary{positioning,decorations.pathreplacing,external,calc,shapes.multipart}

\pgfdeclarelayer{background}
\pgfsetlayers{background,main}

\usepackage{clrscode3e}

\def\EndIf{\End\li\kw{end if} }
\def\EndWhile{\End\li\kw{end while} }

\usepackage{float}
\floatstyle{ruled}
\newfloat{algorithm}{tbp}{loa}
\floatname{algorithm}{Algorithm}
\usepackage[
 			final,
 			unicode=true, 
            bookmarks=true,
            bookmarksnumbered=true,
            bookmarksopen=true,
            breaklinks=true,
            pdfborder={0 0 0}, % no border around links
            colorlinks=false   % do not make links appear red
]{hyperref}

\newcommand{\qed}{\hfill$\square$}
\let\oldendproof\endproof
\renewcommand{\endproof}{\qed\oldendproof}

\vfuzz \hfuzz	

\allowdisplaybreaks[3]

\newcommand\weakemph[1]{\textsl{#1}}
 
\newdimen\makeboxdimen
\newcommand\makeboxlike[3][l]{%
\setbox0=\hbox{#2}%
\global\makeboxdimen=\wd0%
\setbox1=\hbox{\makebox[\makeboxdimen][#1]{%
\makebox[0pt][#1]{#3}%
}}%
\ht1=\ht0%
\dp1=\dp0%
\box1%	
}

\newcommand\plaincenter[1]{%
	\mbox{}\hfill#1\hfill\mbox{}%
}

\newcounter{inlineenum}

\newcommand*\ie{\mbox{i.\hspace{.2ex}e.}}
\newcommand*\eg{\mbox{e.\hspace{.2ex}g.}}

\newcommand\E{\mathop{\mbox{$\mathbb{E}$}}\nolimits}
\newcommand\given{\;|\;}

\newcommand\R{\mathbb R}
\newcommand\N{\mathbb N}
\newcommand\Z{\mathbb Z}

\renewcommand\C{\mathbb{C}}
\newcommand\Oh{O}
\def\.{\mskip1mu}

\newcommand{\Prob}{\ensuremath{\mathbb{P}}}

\newcommand{\ui}[2]{%
	\mathchoice{%display
		#1^{(#2)}%
	}{%text
		#1^{%
			\smash{\raisebox{.75pt}{\mbox{\tiny$($}}}%
			#2%
			\smash{\raisebox{.75pt}{\mbox{\tiny$)$}}}%
		}%
	}{%script
		#1^{\smash{(}\rule[-.75pt]{0pt}{2.75pt}#2\smash{)}}%
	}{%scriptscript
		#1^{(#2)}%
	}%
}

\newcommand{\vect}[1]{\boldsymbol{\mathbf{#1}}}

\newcommand\eqdist{	
	\mathchoice{%display
		\mathrel{\overset{\raisebox{0ex}{$\scriptstyle \cal D$}}=}%
	}{%text
		\mathrel{\like{=}{%
			\overset{\raisebox{-1ex}{$\scriptscriptstyle \cal D$}}=%
		}}%
	}{%script
		\mathrel{\overset{\cal D}=}%
	}{%scriptscript
		\mathrel{\overset{\cal D}=}%
	}%
}

\newcommand\ppe{\phantom{=}}

\newcommand\like[3][c]{\makeboxlike[#1]{\ensuremath{#2}}{\ensuremath{#3}}}

\newcommand\uniform{\mathcal U}

\newcommand\dirichlet{\mathrm{Dir}}
\newcommand\gammadist{\mathrm{Gamma}}

\newcommand\BetaFun{\mathrm B}

\renewcommand\given{\mathbin{\mid}}
\newcommand\harm[1]{\ensuremath{H_{#1}}}

\newcommand\ce{\colonequals}

\newcommand\rel[1]{\mathrel{\:{#1}\:}}
\newcommand\wrel[1]{\mathrel{\;{#1}\;}}
\newcommand\wwrel[1]{\mathrel{\;\;{#1}\;\;}}
\newcommand\bin[1]{\mathbin{\:{#1}\:}}

\newcommand\eqwithref[2][c]{%
	\relwithref[#1]{#2}{=}%
}
\newcommand\relwithref[3][c]{%
	\mathrel{\underset{\mathclap{\makebox[\widthof{$=$}][#1]{\scriptsize\wref{#2}}}}{#3}}%
}
\newcommand\relwithtext[3][c]{%
	\mathrel{\underset{\mathclap{\makebox[\widthof{$=$}][#1]{\scriptsize#2}}}{#3}}%
}

\newcommand\toll[2][]{%
	\ensuremath{%
	\ifthenelse{\equal{#1}{}}{%
		T_{\!#2}% maybe use raisebox to lower?
	}{%
		T_{\!#2}({#1})%
	}}%
}

\newcommand\istoll[2][]{%
	\ensuremath{%
	\ifthenelse{\equal{#1}{}}{%
		\iscost_{\!#2}% maybe use raisebox to lower?
	}{%
		\iscost_{\!#2}({#1})%
	}}%
}

\newcommand\insertsortcost{W}
\newcommand\iscost{\insertsortcost}

\newcommand\bytecodes{\mathit{BC}}
\newcommand\branchmisses{\mathit{BM}}

\newcommand\tp{\alpha}
\newcommand\kp{\kappa}
\newcommand\rf[2]{#1^{\overline{#2}}}

\newcommand\dirichletExpectation[2]{%
	\ensuremath{%
		\E_{D(#2)}\!\left[ #1 \right]%
	}%
}

\newcommand{\geoDirichletExp}[3][1]{\ui{\gamma_{#2,\.#3}}{#1}}

\newcommand\discreteEntropy[1][\vect t]{%
	\ensuremath{\mathchoice{%display
		{\mathcal{H}} \ifthenelse{\equal{#1}{}}{ }{ (#1) }
	}{%text
		{\mathcal{H}} \ifthenelse{\equal{#1}{}}{ }{ (#1) }
	}{%script
		{\mathcal{H}} \ifthenelse{\equal{#1}{}}{ }{ (#1) }
	}{%scriptscript
		{\mathcal{H}} \ifthenelse{\equal{#1}{}}{ }{ (#1) }
	}}\xspace%
}
\newcommand\contentropy[1][\vect\tau]{%
	\ensuremath{\mathchoice{%display
		{\mathcal{H}^*} \ifthenelse{\equal{#1}{}}{ }{ (#1) }
	}{%text
		{\mathcal{H}^*} \ifthenelse{\equal{#1}{}}{ }{ (#1) }
	}{%script
		{\mathcal{H}}^* \ifthenelse{\equal{#1}{}}{ }{ (#1) }
	}{%scriptscript
		{\mathcal{H}}^* \ifthenelse{\equal{#1}{}}{ }{ (#1) }
}}\xspace%
}

\newcommand\arrayA{%
	\ensuremath{\mathchoice{%display
		\smash{\raisebox{-.2pt}{\scalebox{1.25}[1.18]{$\mathtt{A}$}}}%
	}{%text
		\smash{\raisebox{-.2pt}{\scalebox{1.25}[1.18]{$\mathtt{A}$}}}%
	}{%script
		\smash{\raisebox{-.2pt}{\scalebox{1.25}[1.18]{$\scriptstyle\mathtt{A}$}}}%
	}{%scriptscript
		\smash{\raisebox{-.2pt}{\scalebox{1.25}[1.18]{$\scriptscriptstyle\mathtt{A}$}}}%
	}}\xspace%
}

\newcommand\bmprob[2][]{%
	\ifthenelse{\equal{#1}{}}{%
		\ifthenelse{\equal{#2}{}}{%
			\mathbb{P}_{\mkern-3mu\scriptscriptstyle\textrm{BM}}%
		}{%
			\ui { \mathbb{P}_{\mkern-3mu\scriptscriptstyle\textrm{BM}} } {#2}%
		}%
	}{%
		\ifthenelse{\equal{#2}{}}{%
			\mathbb{P}_{\mkern-3mu\scriptscriptstyle#1}%
		}{%
			\ui { \mathbb{P}_{\mkern-3mu\scriptscriptstyle#1} } {#2}%
		}%
	}%
}

\newcommand\bmonebit{\text{1-bit}}
\newcommand\bmtwobitsc{\text{2-bit\,sc}}
\newcommand\bmtwobitfc{\text{2-bit\,fc}}

\newcommand\steadystatemissrate[1]{%
	\ifthenelse{\equal{#1}{}}{%
		f
	}{%
		f_{#1}
	}
}
\newcommand{\steadystatemissrateonebit}{\steadystatemissrate{\bmonebit}}
\newcommand{\steadystatemissratetwobitsc}{\steadystatemissrate{\bmtwobitsc}}
\newcommand{\steadystatemissratetwobitfc}{\steadystatemissrate{\bmtwobitfc}}

\newcommand\btprob[1]{%
	\ui { \mathbb{P}_{\mkern-3mu\scriptscriptstyle\textrm{taken}} } {#1}%
}

\usepackage{manfnt}
\newcommand\warning{\raisebox{-.7ex}{\scalebox{.66}{\textdbend}}}

\newcommand*\generalYaros[2]{%
	\ensuremath{\mathrm{YQS}_{#1}^{#2}}\xspace%
}
\newcommand*\generalClassic[2]{%
	\ensuremath{\mathrm{CQS}_{#1}^{#2}}\xspace%
}

\newcommand*\generalYarostM{%
	\generalYaros{\vect t}{\isthreshold}%
}
\newcommand*\generalClassictM{%
	\generalClassic{\vect t}{\isthreshold}%
}

\newcommand*\isthreshold{\ensuremath{\mathnormal{w}}\xspace}

\let\oldparagraph\paragraph
\makeatletter%
\renewcommand\paragraph{%
    \@ifstar{\myparagraphStar}{\myparagraphNoStar}%
}
\makeatother%
\newcommand\myparagraphStar[1]{%
	\oldparagraph*{#1.}%
}
\newcommand\myparagraphNoStar[2][]{%
	\ifthenelse{\equal{#1}{}}{%
		\oldparagraph[#2]{#2.}%
	}{%
		\oldparagraph[#1]{#2.}%
	}%
}

\colorlet{symmetriccolor}{black!10}

%% file: branch-misses-Introduction.tex
\textsl{Quicksort (QS)} is one of the most intensively used sorting algorithms, \eg,  
as the default sorting method in the 
standard libraries of C, C++, Java and Haskell. 
Classic Quicksort (CQS) uses one element of the input as \textsl{pivot}~$P$ 
according to which the input is 
\emph{partitioned} into the elements smaller than~$P$ and the ones larger than~$P$, 
which are then sorted recursively by the same procedure.

The choice of the pivot is essential for the efficiency of Quicksort. 
If we always use the smallest or largest element of the (sub-)array, 
quadratic runtime results,
whereas using the median gives an (asymptotically) 
comparison-optimal sorting algorithm.
Since the precise computation of the median is too expensive, 
\textsl{sampling strategies} have been invented: 
out of a sample of $k$~randomly selected elements of the input, 
a certain \textsl{order statistic} is selected
as the pivot\,---\,the so-called \textsl{median-of-three} 
strategy is one prominent example of this approach.

In theory, Quicksort can easily be generalized to split the input into 
$s\ge2$ partitions around $s-1$ pivots. (CQS corresponds to $s=2$). 
However, the implementations of \citeauthor{Sedgewick1975} and others
did not perform as well in running time experiments as 
classic single-pivot Quicksort \citep{Sedgewick1975}; 
it was common belief that the overhead of 
using several pivots is too large in practice. 
In 2009, however, Vladimir Yaroslavskiy proposed a new dual-pivot variant of 
Quicksort which surprisingly outperformed the 
highly engineered classic Quicksort of Java~6, 
which then lead to its replacement in Java~7 by 
Yaroslavskiy's dual-pivot Quicksort (YQS). 

An analytic explanation of the superiority of YQS was lacking at that time
(and is still open to some extent).
We showed that YQS indeed saves $5\,\%$ of comparisons 
(for random pivots); but also that it needs 
$80\,\%$ \emph{more} swaps than CQS 
\citep{Wild2012,Wild2013Quicksortarxiv}. 
Arguably, these traditional cost measures do not explain the running times well
and it is likely that features of 
modern CPUs like \emph{memory hierarchies} and \emph{pipelines} are the prime reason for the 
algorithm's efficiency. 
This would be in accordance with the aforementioned results 
in which multi-pivot Quicksort was less efficient, 
as the experiments were done in a time 
when computers simply did not have such features.

In this paper we address the effects of pipelines, 
which are used to speed up execution as follows: 
Inside the CPU, machine instructions are split into \emph{phases} like 
``fetching the instruction'', 
``decoding and loading data'', 
``executing the instruction'', and ``writing back results''.
Each phase takes one CPU cycle.
Modern CPUs execute different phases in parallel, \ie, 
if there are $L$ phases, one can execute $L$ instructions at once, 
each in a different phase, resulting in a speed-up of~$L$. 

The downside of this idea, however, comes with \emph{conditional jumps}. 
For those, the CPU will have to decide the 
outcome \emph{before} it has actually been computed\,---\,%
otherwise it could not go on with the first phases of the subsequent commands.
To cope with that, several \textsl{branch prediction schemes} have been invented,
which try to \emph{guess} the actual outcome.
In the simplest case  each branch (conditional jump) is marked 
``probably taken'' or ``probably not taken'' at compile time
and the CPU acts accordingly. 
As branch outcomes most often depend on data not known at compile time,
this strategy is quite limited.
In order to \emph{adapt} predictions to actually observed 
behavior special hardware support is needed. 
For the so-called \textsl{1-bit predictor}, the CPU stores for each 
branch in the code whether or not it was taken the last time it was executed 
and assumes the same behavior for the future.
In a \textsl{2-bit prediction} scheme, the CPU has to make a wrong prediction twice  
before switching. 
Such simple schemes have been used by the first CPUs with 
pipelining. 
Modern microprocessors implement more sophisticated heuristics~\cite{AgnerMicroarchitecture},
which try to recognize common patterns in branching behavior.
As they seem too intricate for precise analysis and 
are probably \emph{inferior} to the simple schemes in Quicksort%
\footnote{%
	There are no branch patterns in Quicksort
	(cf.\ \wref{sec:dirichlet-vectors-and-independence})!
}
we focus on the basic predictors from above.

In case of a false prediction 
(\textsl{branch misprediction} or \textsl{branch miss}, briefly BM)
the CPU has to \emph{undo} all erroneously executed steps (phases) 
and load the correct instructions instead. 
A branch miss is thus a costly event. 
As an extreme example, \citet{Kaligosi2006branch} observed on a 
Pentium~4 Prescott CPU 
(a processor with an extremely long pipeline and thus a high cost per BM) 
that the running time penalty of a BM is so high that 
a very skewed pivot choice outperformed the 
typically optimal median pivot, 
even though the latter leads to much less executed instructions in total.
The effect was not reproducible, however, on the slightly different 
Pentium~4 Willamette \cite{Biggar2008}.
Here two effects counteract: 
a biased pivot makes branches easier to predict but 
also gives unbalanced subproblem sizes. 
\Citet{brodal2005tradeoffs} have shown that this is a general trade-off
in comparison-based sorting:
One can only save comparisons at the price of additional branch misses.

Differences in the number of branch misses might be an 
explanation for the superiority of YQS and will 
thus be analyzed in this paper in connection with pivot sampling. 
We will also reconsider branch misses in classic Quicksort by
continuing the analyses of 
\citet{Kaligosi2006branch} and \citet{Biggar2008},
we present precise leading terms%
\footnote{%
	By our discussion in \wref{sec:dirichlet-vectors-and-independence}, 
	we also answer a question posed by
	\citet{Kaligosi2006branch} in their footnote~2, where they
	ask for an argument to make their heuristic analysis\,---\,assuming a 
	constant probability $\alpha$ for an element to be small\,---\,rigorous.
	We find it most appropriate to answer a footnote question also in a footnote,
	even though we consider the settlement of this open problem quite
	noteworthy.
}
and explicitly address pivot sampling for finite sample sizes~$k$.

We find that CQS and YQS cause roughly the same number of branch misses and 
hence pipelining effects are \emph{not}
a likely explanation for the superiority of YQS. 
Even if this result in isolation appears negative, it still entails valuable new insights:
it provides further evidence for the hypothesis of \citet{Kushagra2014} that 
it is the impact of memory hierarchies in modern computers that renders YQS faster.

\section{Generalized Quicksort}
\label{sec:generalized-quicksort}

In this section, we review classic Quicksort and 
Yaroslavskiy's dual-pivot variant and introduce the 
generalized pivot-sampling method.

\subsection{Generalized Pivot Sampling}
\label{sec:general-pivot-sampling}

\begin{algorithm}
	\small
	\vspace{-1ex}
	\def\pind{\id{i_p}}
	\begin{codebox}
		\Procname{$\proc{PartitionSedgewick}(A,\id{left},\id{right},p)$}
		\zi \Comment Assumes $\id{left} \le \id{right}$ and a sentinel $A[0]=-\infty$.
		\zi \Comment Rearranges \arrayA such that with return value $\pind$
				holds  \\[.5ex]
			\qquad \qquad
				$\begin{cases}
	 				\forall \: \id{left} \le j < \pind,		& \arrayA[j] \le p; \\
	 				\forall \: \pind \le j \le \id{right},		& \arrayA[j] \ge p .
				\end{cases}$
				\\[-1ex]
		\li $k\gets \id{left}-1$;
			\quad $g\gets \id{right}$
		\li \kw{do}\Do
		\li		\kw{do} $k\gets k+1$ \While $A[k] < p$ \kw{end while}
						\label{lin:classic-comp-1}
		\li		\kw{do} $g\gets g-1$ \While $A[g] > p$ \kw{end while}
						\label{lin:classic-comp-2}
		\li		\If $g>k$ \kw{then} Swap $A[k]$ and $A[g]$ \kw{end if} 
						\label{lin:classic-swap}
			\End
		\li	\While $g > k$ \label{lin:classic-outer-loop-branch}
		\li	\Return $k$
	\end{codebox}
	\vspace{-1ex}
	\plaincenter{\textbf{Invariant:}}\\[1ex]
	\plaincenter{
		\begin{tikzpicture}[
			yscale=0.5, xscale=0.5,
			baseline=(ref.south),
			every node/.style={font={\small}},
			semithick,
		]	
		
		\draw (-.75,0) -- ++(13,0) -- ++(0,1) -- ++(-13,0) -- cycle;
		\node at (-.5,-0.5) {$\id{left}$};
		\node at (12,-0.5) {$\id{right}$};

		\node at (2-.375,0.5) {$\le P$};
		\draw (4,1) -- ++ (0,-1);
		\node at (4.3,-0.4) {$k$};
		
		\node at (10.125,0.5) {$\ge P$};
		\draw (8,1) -- ++ (0,-1);
		\node at (7.7,-0.4) {$g$};

		\node[below] at (7.7,-0.5) {$\leftarrow$};
		\node[below] at (4.3,-0.5) {$\rightarrow$};
		
		\node[inner sep=0pt] (ref) at (6,0.5) {?};
		\end{tikzpicture}%
	}
	\caption{\strut%
		Classic Crossing-Pointer Partitioning.
	}
	\label{alg:partition-sedgewick}
\end{algorithm}

% Yaroslavskiy Partitioning algorithm
\begin{algorithm}
	\small
	\vspace{-1ex}
	\def\pind{\id{i_p}}
	\def\qind{\id{i_q}}
	\begin{codebox}
		\Procname{$\proc{PartitionYaroslavskiy}\,(\arrayA,\id{left},\id{right},p,q)$}
		\zi \Comment Assumes $\id{left} \le \id{right}$. 
		\zi \Comment Rearranges \arrayA s.\,t.\ with return value $(\pind,\qind)$
				holds \\[.5ex]
					\qquad\qquad
				$\begin{cases}
	 				\forall \: \id{left} \le j \le \pind,		& \arrayA[j] < p; \\
	 				\forall \: \like[l]{\id{left}}{\pind} < j < \qind,		& p \le \arrayA[j] \le q; \\
	 				\forall \: \like[l]{\id{left}}{\qind} \le j \le \id{right},		& \arrayA[j] \ge q .
				\end{cases}$
				\\[-1ex]
		\li $\ell\gets \id{left}$; 
		 	\quad $g\gets \id{right}$; 
		 	\quad $k\gets \ell$ \label{lin:yaroslavskiy-init-l-g-k} 
		\li	\While $k\le g$  \label{lin:yarosavskiy-outer-loop-branch}
		\li	\Do
				\If $\arrayA[k] < p$ \label{lin:yaroslavskiy-comp-1}
		\li		\Then
					Swap $\arrayA[k]$ and $\arrayA[\ell]$ \label{lin:yaroslavskiy-swap-1}
		\li			$\ell\gets \ell+1$ \label{lin:yaroslavskiy-l++-1}
		\li		\Else 
		\li			\If $\arrayA[k] \ge q$ \label{lin:yaroslavskiy-comp-2}
		\li			\Then
						\While $\arrayA[g] > q$ and $k<g$ \label{lin:yaroslavskiy-comp-3} 
		\li				\Do 
							$g\gets g-1$ 
						\EndWhile
		\li				\If $\arrayA[g] \ge p$ \label{lin:yaroslavskiy-comp-4}
		\li				\Then
							Swap $\arrayA[k]$ and $\arrayA[g]$ \label{lin:yaroslavskiy-swap-2}
		\li				\Else
		\li					Swap $\arrayA[k]$ and $\arrayA[g]$ \label{lin:yaroslavskiy-swap-3a}
		\li					Swap $\arrayA[k]$ and $\arrayA[\ell]$ \label{lin:yaroslavskiy-swap-3b}
		\li					$\ell\gets \ell+1$ \label{lin:yaroslavskiy-l++-2}
						\EndIf
		\li				$g\gets g-1$ \label{lin:yaroslavskiy-g--}
					\EndIf
				\EndIf
		\li		$k\gets k+1$ \label{lin:yaroslavskiy-k++}
			\EndWhile \label{lin:yaroslavskiy-end-while}
		\li	\Return $(\ell-1, g+1)$
	\end{codebox}
	\vspace{-1ex}
	\plaincenter{\textbf{Invariant:}}\\[1ex]
	\plaincenter{
		\begin{tikzpicture}[
			yscale=0.5, xscale=0.5,
			baseline=(ref.south),
			every node/.style={font={\small}},
			semithick,
		]	
		
		\draw (-.75,0) -- ++(13,0) -- ++(0,1) -- ++(-13,0) -- cycle;
%		\node at (-.375, .5) {$P$} ;
%		\node at (13.375, .5) {$Q$} ;
%		\draw (0,0) -- ++(0,1);
%		\draw (13,0) -- ++(0,1);
		\node at (-.5,-0.5) {$\id{left}$};
		\node at (12,-0.5) {$\id{right}$};

		\node at (.75,0.5) {$< P$};
		\draw (2.25,1) -- ++ (0,-1);
		\node at (2.55,-0.4) {$\ell$};
		
		\node at (11.175,0.5) {$\ge Q$};
		\draw (10,1) -- ++ (0,-1);
		\node at (9.7,-0.4) {$g$};
		
		\node at (4.5,0.5) {$P\le \circ\le Q$};
		\draw (6.5,1) -- ++(0,-1);
		\node at (6.8,-0.4) {$k$};
	
		\node[below] at (9.7,-0.5) {$\leftarrow$};
		\node[below] at (2.55,-0.5) {$\rightarrow$};
		\node[below] at (6.8,-0.5) {$\rightarrow$};
		
		\node[inner sep=0pt] (ref) at (8.25,0.5) {?};
		\end{tikzpicture}%
	}
	\caption{\strut%
		Yaroslavskiy's dual-pivot partitioning.
	}
\label{alg:partition-yaroslavskiy}
\end{algorithm}

For the one-pivot case, our pivot selection process is 
declaratively specified as follows:
for $\vect t = (t_1,t_2) \in \N^2$ a fixed parameter,
choose a random sample $\vect V = (V_1,\ldots,V_k)$ of size 
$k = k(\vect t) \ce t_1+t_2+1$ from the input.
If we denote the \emph{sorted} sample by
$
	V_{(1)} \le V_{(2)} \le \cdots \le V_{(k)}
$,
we choose the pivot $P \ce V_{(t_1+1)}$ such 
that it divides the sorted sample into regions of respective sizes $t_1$ and $t_2$:
\begin{multline*}
	\underbrace{V_{(1)} \ldots V_{(t_{1})}}
			_{t_{1}\,\mathrm{elements}}
	\wrel\le
	\underbrace{V_{(t_{1}+1)}} _ {=P}
	\wrel\le
	\underbrace{V_{(t_{1}+2)} \ldots V_{(k)}}
			_{t_{2}\,\mathrm{elements}}
	\,.
\end{multline*}
The choice $\vect t = 0$ corresponds to no sampling at all.

For dual-pivot Quicksort, we have to choose two pivots instead of just one.
Therefore, with $\vect t \in \N^3$, we choose a sample 
of size $k = t_1+t_2+t_3+2$ and
take $P \ce V_{(t_1+1)}$ and $Q\ce V_{(t_1+t_2+2)}$ as pivots, which
divide the sorted sample into three regions of respective sizes $t_1$,
$t_2$ and~$t_3$.
Note that by definition, $P$ is the small(er) pivot and $Q$ is the
large(r) one.

\subsection{Classic Crossing-Pointer Partitioning}

The classic implementation of (single-pivot) Quicksort dates back to 
\citeauthor{Hoare1961b}'s original publication of the algorithm \cite{Hoare1961b}, 
later refined and popularized by \citet{Sedgewick1978}.
Detailed code for the partitioning procedure is given in \wref{alg:partition-sedgewick}.
It uses two ``crossing pointers'', $k$ and $g$, starting at the left resp.\ 
right end of the array and moving towards each other until they meet.
When a pointer reaches an element that does not belong to this pointer's 
partition, it stops. Once both have stopped, 
the out-of-order pair is exchanged and the pointers go on.
The array is thereby kept invariably in form shown below the code in 
\wref{alg:partition-sedgewick}.

\subsection{Yaroslavskiy's Dual-Pivoting Method}

Yaroslavskiy's partitioning method also consists of two
indices, $k$ and~$g$, that start at the left resp.\ right end and
scan the array until they meet. 
Additionally, however, a third index~$\ell$ ``lags'' behind~$k$ to further 
divide the region left of~$k$.
Detailed pseudocode is given in \wref{alg:partition-yaroslavskiy}, where
we also give the invariant maintained by the algorithm (below the code).

\subsection{From Partitioning to Sorting}

When partitioning is finished, $k$ and $g$ have met and thus 
for CQS, divide the array into two ranges, 
containing the elements smaller resp.\ larger than $P$; 
for YQS, $\ell$ and $g$
induce \emph{three} ranges, containing the elements that are
smaller than $P$, between $P$ and $Q$ 
resp.\ larger than $Q$;
(see also the invariants, when the ``?''-area has vanished).
Those regions are then sorted recursively, independently of each other.

We omit detailed pseudocode for the full sorting procedures since 
the only contributions to the leading term of costs come form partitioning.%
\footnote{%
	Note, however, that some care is needed to preserve randomness in
	recursive calls, which is in turn crucial to set up recurrence equations.
	See Appendix~B of \citep{NebelWild2014} for further details on how to
	achieve this for Generalized Yaroslavskiy Quicksort.
}
(Recall that we consider $k=\Oh(1)$ constant.)

To implement generalized pivot sampling, $k$ elements from the array are chosen
and the needed order statistic(s) are selected from this sample.
Note that after sampling, we know for sure to which partition the sample elements belong, 
so we can exclude them from partitioning:
the partitioning methods as they are given here are only applied to the ``ordinary'' 
elements, \ie, the $n-k$ elements that have not been part of the sample.

For subarrays with at most \isthreshold elements, we switch to Insertionsort 
(cf.\ \citep{Sedgewick1978} resp.\ \citep{NebelWild2014}),
where \isthreshold is constant and at least $k$.
The resulting algorithms, Generalized Classic resp.\ Yaroslavskiy Quicksort with pivot
sampling parameter $\vect t$ and Insertionsort threshold
\isthreshold, are henceforth called $\generalClassictM$ and $\generalYarostM$.

\section{Notation and Preliminaries}
\label{sec:notation}

To unify formal notation, we denote by $s\ge 2$ 
the number of partitions produced in one step, 
\ie, $s=2$ for CQS, $s=3$ for YQS, and no other
values for $s$ are considered in this paper.
Recall that $s-1$ pivots are needed for partitioning into $s$ parts.

We write vectors in bold font, for example 
$\vect t=(t_1,\ldots,t_{s})$.
For concise notation, we use expressions like $\vect t + 1$ to mean
\emph{element-wise} application, \ie, $\vect t + 1 = (t_1+1,\ldots,t_{s}+1)$.
By $\dirichlet(\vect\alpha)$, we denote a random variable with \textsl{Dirichlet
distribution} and shape parameter
$\vect\alpha = (\alpha_1,\ldots,\alpha_d) \in \R_{>0}^d$.
Likewise, $\uniform(a,b)$ is a random variable uniformly distributed in the
interval $(a,b)$.
We use ``$\eqdist$'' to denote equality in distribution.

As usual for the average case analysis of sorting algorithms, we assume the
\textsl{random permutation model}, \ie, all elements are
different and every ordering of them is equally likely.
The input is given as array $\arrayA$ of length $n$ and we denote the initial
entries of $\arrayA$ by $U_1,\ldots,U_n$.
We further assume that $U_1,\ldots,U_n$ are
i.\,i.\,d.\ uniformly $\uniform(0,1)$ distributed;
as their ordering forms a random permutation \citep{mahmoud2000sorting}, this
assumption is without loss of generality.

For YQS, we call  an element \emph{small}, \emph{medium}, or \emph{large} if
it is smaller than~$P$, between $P$ and $Q$, or larger than~$Q$, respectively;
for CQS, elements are either smaller than $P$ or larger.

%% file: branch-misses-analysis.tex
\section{Generalized Quicksort Recurrence}
\label{sec:analysis-recurrence}

For $\vect t \in \N^s$ and $\harm{n}$ the $n$th harmonic number, we define the
\textsl{discrete entropy} $\discreteEntropy[]$ as
\begin{align}
\label{eq:discrete-entropy}
		\discreteEntropy[] \wrel{=} \discreteEntropy[\vect t]
	&\wwrel=
		\sum_{l=1}^s \frac{t_l+1}{k+1} 
				(\harm{k+1} - \harm{t_l+1})
	\;. 
\end{align} 
In the limit $k\to\infty$, such that the ratios ${t_l}/k$ 
converge to constants~$\tau_l$, \discreteEntropy coincides with the
\weakemph{entropy function} $\contentropy[]$ of information theory: 
\begin{align}
\label{eq:limit-entropy}
		\discreteEntropy
	&\wwrel\sim
		- \sum_{l=1}^s \tau_l \bigl( \ln(t_l+1) - \ln(k+1) \bigr)
\\*	&\wwrel\sim	  - \sum_{l=1}^s \tau_l \ln(\tau_l)
	\wwrel{\equalscolon} \contentropy
	\;.
\end{align}
The first step follows from the asymptotic equivalence
$\harm{n} \sim \ln(n)$ as $n\to\infty$.

\begin{theorem}[Quicksort Recurrence]
\label{thm:leading-term-expectation-hennequin}
~\\
	The total expected costs for sorting a random permutation with Quicksort using 
	a partitioning method that incurs expected costs 
	$\E[T_n]$ of the form $\E[T_n] = an+\Oh(1)$
	to produce $s$ partitions and
	whose $s-1$ pivots are chosen by generalized pivot sampling with parameter 
	$\vect t\in\N^s$ are asymptotically
	\smash{$
			\sim \frac{a}{\discreteEntropy[]} \, n \ln n
	$},
	where $\discreteEntropy[]$ is given by \wref{eq:discrete-entropy}.
\end{theorem}

\wref{thm:leading-term-expectation-hennequin} has first been proven by
\citet[Proposition~III.9]{hennequin1991analyse} using arguments on the
Cauchy-Euler differential equations for the corresponding generating function of costs.
A more concise and elementary proof using 
\citeauthor{Roura2001}'s \textsl{Continuous Master Theorem}~\citep{Roura2001} 
is given in the appendix of \citep{NebelWild2014}.

\needspace{5\baselineskip}
\section{Branch Mispredictions}
\label{sec:branch-mispredictions}

Thanks to \wref{thm:leading-term-expectation-hennequin} we can easily make the transition
from expected \emph{partitioning} costs to the corresponding \emph{overall} costs, so
in the following, we can focus on the first partitioning step only.

For our Quicksort variants, there are two different types of branches: 
those which correspond to a key comparison and all others (loop headers etc.).
It turns out that all non-comparison branches are highly predictable, 
meaning that any reasonable prediction scheme will only incur a constant number of
mispredictions (per partitioning step):

In CQS we have one backward branch at the end of the outer loop with condition ``$g > k$''
(\wref{lin:classic-outer-loop-branch} of \wref{alg:partition-sedgewick})
and an if-statement with the same condition
(\wref{lin:classic-swap}), both are violated exactly once, then we exit.
Similar in YQS, there is an outer loop branch with ``$k\le g$'' 
(\wref{lin:yarosavskiy-outer-loop-branch} in \wref{alg:partition-yaroslavskiy}) 
and an inner one with ``$k<g$'' (\wref{lin:yaroslavskiy-comp-3}),
which again fail at most once.
For the leading term of the number of branch misses (BM), we can therefore focus on the 
comparison-based branches
(two in CQS and four in YQS).

All prediction schemes analyzed in this paper are \emph{local} in the sense that 
for any branch instruction in the code, the prediction of its next outcome only depends on 
formerly observed behavior of this very branch instruction.
As a consequence, predictions of one branch instruction in the code are \emph{independent} of 
the history of any \emph{other} branch instruction.
Note however that the histories themselves might be highly inter-dependent in general,
which complicates the analysis of branch misses.

\subsection{Dirichlet Vectors and Conditional Independence}
\label{sec:dirichlet-vectors-and-independence}

Recall that our input consists of i.\,i.\,d.\ $\uniform(0,1)$
variables.
If, for CQS, we \emph{condition} on the pivot value, \ie, consider $P$ 
fixed, an ordinary element $U$ is small, if $U \in (0,P)$,
and large if $U \in (P,1)$.
If we call the (random) lengths of these two intervals $\vect D = (D_1,D_2) = (P,1-P)$,
then $D_1$ and $D_2$ are the \emph{probabilities} for an element to be small resp.\ large.
For YQS, we condition on both $P$ and $Q$ and correspondingly get $\vect D = (D_1,D_2,D_3) = 
(P,Q-P,1-Q)$ as the probabilities for small, medium resp.\ large elements:

\medskip
\noindent\plaincenter{%
	\def\r{1.5pt}
	\begin{tikzpicture}[
		every node/.style={font=\footnotesize},
	]
		\useasboundingbox (0,-0.4) rectangle (5,0.5) ;
		
		\draw[|-|] (0,0) node[below=.5ex] {$0$}  -- (5,0) node[below=.5ex] {$1$};
		\filldraw (1,0) circle (\r) node[below=.5ex] {$P$};
		\filldraw (3.5,0) circle (\r) node[below=.5ex] {$Q$};
		\begin{scope}[
				yshift=1.5ex,<->,
				shorten >=.3pt,shorten <=.3pt,
				fill=white,inner sep=1pt,
		]	
			\draw (0,0)   -- node[fill] {$D_1$} (1,0) ;
			\draw (1,0)   -- node[fill] {$D_2$} (3.5,0) ;
			\draw (3.5,0) -- node[fill] {$D_3$} (5,0) ;
		\end{scope}
	\end{tikzpicture}%
}
\medskip

The random variable $\vect D \in [0,1]^s$ is a vector of \textsl{spacings} 
induced by order statistics from a sample of $\uniform(0,1)$ variables in the unit interval,
which is known to have a \weakemph{Dirichlet}
$\dirichlet(\vect t + 1)$ distribution 
(\wref{pro:spacings-dirichlet-general-dimension} in 
\wref{app:distributions}).

The vital observation for our analysis is that the probabilities $\vect D$ 
of the class of an element $U$ are \emph{independent} of all other (ordinary) elements!
For the comparisons during partitioning, this implies that their \emph{outcomes}
are i.\,i.\,d.; precisely speaking: 
\textsl{the sequence of (binary) random variables that correspond to the 
outcomes of all executions (in one partitioning step) of the 
key comparison at \textit{one comparison location}
are i.\,i.\,d.,}  and their distribution only depends
on the pivot \emph{values} (via $\vect D$).
The outcome of a comparison does neither depend on the position in the array, 
nor on the number of, say, small elements that we have already seen.

As the outcomes for different comparison locations are always independent, so are the corresponding
branch histories and predictions.
Conditioning on $\vect D$, we can therefore count the branch misses separately for
all comparison locations.

\subsection{Probabilities of Branches in CQS}

CQS has two comparison-based branch locations: 
$\ui C{c1}$ and $\ui C{c2}$ in the two inner loops
(\wref[lines]{lin:classic-comp-1} and~\ref*{lin:classic-comp-2} in \wref{alg:partition-sedgewick}).
The first one jumps back to the loop header if $\arrayA[k] < P$, the second one if
$\arrayA[g] > P$.
Conditional on $\vect D = (D_1,D_2)$, the two branches are executed 
$\ui {C_n}{c1} = D_1\.n + \Oh(1)$ resp.\ $\ui {C_n}{c2} = D_2\.n + \Oh(1)$ 
times in expectation (in the first partitioning step) and 
each time, they are taken i.\,i.\,d.\ with probability
$\btprob{c1} = D_1$ resp.\ $\btprob{c2} = D_2$, 
(the probabilities for the element to be small resp.\ large).

Note that CQS is symmetric:
At both $\ui C{c1}$ and $\ui C{c2}$, we compare an element with $P$, 
so the two branches behave the same w.\,r.\,t.\ branch misses.
It is then convenient to work with the \emph{combined} (virtual) 
comparison location~$\ui Cc$ which is
executed $\ui{C_n}c = \ui{C_n}{c1} + \ui{C_n}{c2} = n+\Oh(1)$ times
with branch probability $\btprob{c} = D_1$.
(Such tricks will fail for the asymmetric YQS.)
Denoting by $\bmprob{c}{}$ the probability for a BM at $\ui Cc$,
the expected number of BMs in one partitioning step of CQS is then simply
\begin{align}
\label{eq:toll-CQS-generic}
		\E[\toll[n]{\branchmisses}] 
	\wwrel= 
		\E[\bmprob{c}{}]\.n+\Oh(1)\;.
\end{align}

\begin{table}
	\small
	\def\C#1{\ui {C_n}{#1}}
	\def\rowsep{2.5pt}
	\def\ed{}
	\def\pd{\vphantom{\ed}}
	\plaincenter{%
	\begin{tabular}{ >{\(}c<{\)} >{\(}l<{\)} >{\(}l<{\)} }
		\toprule
		\multicolumn1{c}{Location}	&	\multicolumn2{c}{Expectation conditional on $\vect D$} \\
		\midrule\\[-2ex]
		\C{c1}	& D_1 & {}\cdot (n + \Oh(1)) \\[\rowsep]
		\C{c2}	& D_2 & {}\cdot (n + \Oh(1)) \\[\rowsep]
		\midrule\\[-2ex]
		\C{y1}	& (D_1 + D_2) & {}\cdot (n + \Oh(1)) \\[\rowsep]
		\C{y2}	& (D_1 + D_2)(D_2 + D_3)\mkern-20mu & {}\cdot (n + \Oh(1)) \\[\rowsep] 
		\C{y3}	& D_3	& {}\cdot (n + \Oh(1)) \\[\rowsep] 
		\C{y4}	& D_3(D_1+D_2) & {}\cdot (n + \Oh(1)) \\[\rowsep]
		\bottomrule
	\end{tabular}%
	}
	\caption{
		The expected execution frequencies of the comparisons locations in CQS and YQS,
		conditional on $\vect D$.
		The full distributions of the frequencies are given in \citep{NebelWild2014},
		the conditional expectations are then easily computed using the lemmas in
		Appendix~C of~\citep{NebelWild2014}.
	}
	\label{tab:cmps-execution-frequencies}
\end{table}

\subsection{Probabilities of Branches in YQS}

For YQS, we have four comparison locations:
$\ui C{y1}$ (\wref{lin:yaroslavskiy-comp-1} in \wref{alg:partition-yaroslavskiy}),
$\ui C{y2}$ (\wref{lin:yaroslavskiy-comp-2}),
$\ui C{y3}$ (\wref{lin:yaroslavskiy-comp-3}) and
$\ui C{y4}$ (\wref{lin:yaroslavskiy-comp-4}).
\wref{tab:cmps-execution-frequencies} lists their execution frequencies.
Note that $\ui C{y2}$ is only reached for elements where $\ui C{y1}$ determined that 
they are not small, so at $\ui C{y2}$, we already know the element is either medium or large.
Similarly, $\ui C{y4}$ only handles elements that $\ui C{y3}$ proved to be non-large.
Recalling that the probability of an (ordinary) element to be small, medium or
large is $D_1$, $D_2$ and $D_3$, respectively, a look at the comparisons yields
\begin{align}
\label{eq:branch-taken-prob-YQS}
		\btprob{y1}
	&\wrel=
		D_2+D_3\,,
&
		\btprob{y2}
	&\wrel=
		\frac{D_2}{D_2+D_3}\,,
\\
		\btprob{y3}
	&\wrel=
		D_1+D_2\,,
&
		\btprob{y4}
	&\wrel=
		\frac{D_1}{D_1+D_2}\;.
\end{align} 
There are no symmetries to exploit in YQS (all comparisons are of different type), 
so we will get different BM probabilities $\bmprob{y1}{},\ldots,\bmprob{y4}{}$ 
for all locations.
The expected number of BM in one partitioning step is
\vspace{-2ex}
\begin{align}
\label{eq:toll-YQS-generic}
		\E[\toll[n]{\branchmisses}]
	\wwrel=
		\sum_{l=1}^4 \E\Bigl[ \ui{C_n}{yl} \cdot \bmprob{yl}{} \Bigr]
	\;.
\end{align}

\subsection{Branch Prediction Schemes}
\label{sec:branch-prediction-schemes}

For each comparison location $\ui Ci$ in the code,
we determined the probability $\btprob{i} = \btprob{i}(\vect d)$ that, 
conditional on $\{\vect D=\vect d\}$, this branch will be taken.
The optimal strategy would then be to predict this branch to be taken iff
$\btprob{i}(\vect D) \ge \frac12$\,---\,which, of course, is not possible since 
we do not know $\vect D$ at runtime.
(This theoretical scheme has been considered by \citet{Kaligosi2006branch} as 
benchmark.)

Actual adaptive prediction schemes try to \emph{estimate} $\btprob{i}$ 
using a limited history of 
previous outcomes and base predictions on their current estimate.
As the CPU has to keep track of this history for many branches, typical memory sizes are 
as low as 1 or 2 bits.
However, in practice it is impossible to reserve even just a single bit of 
branch history for every possible branch-instruction location.
Therefore actual hardware prediction units use \emph{hash tables} of history storage,
which means that all branch instructions whose addresses hash to the same value will 
share one history storage. 
The resulting \textsl{aliasing effects} have typically small, but rather chaotic influences on 
predictions in practice. 
We ignore those to keep analysis tractable.

The behavior of an (idealized) local adaptive prediction scheme
(for a single branch instruction) forms a \textsl{Markov chain} over the states of  
its finite memory of past behavior.
Each state corresponds to a prediction (``taken'' / ``not taken'') and has two successor
states depending on the actual outcome.

As the involved Markov chains have only 2 or 4 states, 
they approach their stationary distributions very quickly.%
\footnote{%
	Numeric experiments show that after 100 iterations the dependence of the state 
	distribution on the initial state is less than $10^{-6}$.
}
For the asymptotic number of branch misses for partitioning a 
large array, we can therefore assume the predictor automaton to have reached its steady state.
Averaging over the states and the outcomes of the next branch, we get the
(expected) BM probability: 
$\bmprob{i}{} = \steadystatemissrate{}(\btprob{i})$ for a (deterministic) function 
$f:[0,1]\to[0,1]$ depending on the prediction scheme that we call its 
\textsl{steady-state miss-rate function}.
(\citet{Biggar2008} call $p \mapsto 1-\steadystatemissrate{}(p)$ 
the \textsl{steady-state predictability}.)
The schemes considered herein only differ in the topology of the Markov chain and
thus in $\steadystatemissrate{}(p)$.

\begin{figure}
	\plaincenter{\includegraphics[width=.66\linewidth]{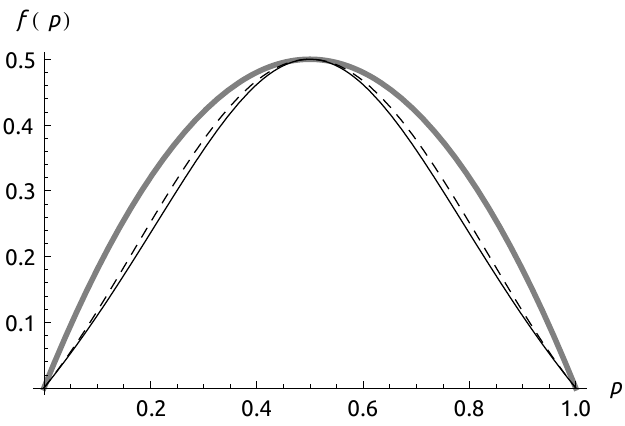}}
	\caption{%
		Comparison of the steady-state miss-rate functions for the 1-bit (thick gray),
		the 2-bit saturating-counter (black) and the 2-bit flip-on-consecutive (dashed) predictors.
		The $x$-axis varies $p$ from $0$ to $1$.
		All functions are symmetric around $1/2$\,---\,%
		where they have a peak with value $1/2$\,---\,%
		and then drop to $0$ as $p$ approaches $0$ or $1$.
		Note that the 2-bit saturating counter leads to the best predictions for any $p$ 
		(as long as all the branch executions are i.\,i.\,d.\ and taken with probability $p$).
	}
	\label{fig:comparison-steady-state-branch-miss}
\end{figure}

\subsubsection*{1-Bit Predictor.}

The 1-bit predictor is arguably the simplest possible adaptive prediction scheme:
it always predicts a branch to behave the same way as the
last time it was executed.
We thus encounter a branch miss for a comparison location $\ui Ci$
if and only if two subsequent executions happen on elements with
different comparison results.
As the branch probabilities for two executions of the same comparison
location in one partitioning step are the same and branching is independent,
we get 
\begin{align}
\label{eq:steady-state-predictability-1bit}
	\steadystatemissrateonebit(p) &\wrel= 2p(1-p) \;.
\end{align}

\subsubsection*{2-Bit Saturating-Counter Predictor.}

\begin{figure}
	\def\taken{\tiny taken}
	\def\nottaken{\tiny not t.}
	\plaincenter{%
	\begin{tikzpicture}[
		st/.style={
			circle split,draw,
			inner sep=1pt,font={\footnotesize},
			minimum size=2em,
			fill=black!5,
		},
		scale=1.5,
		->,>=stealth,
	]
		\draw[overlay,-,dotted,ultra thick,black!20] (-0.5,.6) -- (-0.5,-.6);
		\node[st] (ST) at (-2,0) {$1$ \nodepart{lower}\taken};
		\node[st] (T) at (-1,0) {$2$ \nodepart{lower}\taken};
		\node[st] (NT) at (0,0) {$3$ \nodepart{lower}\nottaken};
		\node[st] (SNT) at (1,0) {$4$ \nodepart{lower}\nottaken};
		\draw (ST) to[out=220,in=270] node[below]{\taken} ++(-.7,0) to[out=90,in=140] (ST) ;
		\draw (ST) to[out=40,in=140] node[above]{\nottaken} (T)  ;
		\draw (T) to[out=40,in=140] node[above]{\nottaken} (NT)  ;
		\draw (NT) to[out=40,in=140] node[above]{\nottaken} (SNT)  ;
		\draw (SNT) to[out=40,in=90] node[above]{\nottaken} ++(.7,0) to[out=270,in=-40] (SNT) ;
		\draw (SNT) to[out=220,in=-40] node[below]{\taken} (NT)  ;
		\draw (NT) to[out=220,in=-40] node[below]{\taken} (T)  ;
		\draw (T) to[out=220,in=-40] node[below]{\taken} (ST)  ;
	\end{tikzpicture}
	}
	\caption{%
		2-bit saturating-counter predictor.
	}
	\label{fig:2-bit-saturating-counter}
\end{figure}

\begin{figure}
	\def\taken{\tiny taken}
	\def\nottaken{\tiny not t.}
	\plaincenter{%
	\begin{tikzpicture}[
			st/.style={
				circle split,draw,
				inner sep=1pt,font={\footnotesize},
				minimum size=2em,
				fill=black!5,
			},
			scale=1.4,
			->,>=stealth,
	]
		\draw[overlay,-,dotted,ultra thick,black!20] (0.9,1.5) -- (0.9,-.5);
		\node[st] (ST) at (0,1) {$1$\nodepart{lower}\taken};
		\node[st] (T) at (0,0) {$2$\nodepart{lower}\taken};
		\node[st] (NT) at (1.8,1) {$3$\nodepart{lower}\nottaken};
		\node[st] (SNT) at (1.8,0) {$4$\nodepart{lower}\nottaken};
		\draw (ST) to[out=220,in=270] ++(-.7,0) to[out=90,in=140]node[above]{\taken} (ST) ;
		\draw (ST) to[bend right] node[left]{\nottaken} (T)  ;
		\draw (T) to[bend right] node[right]{\taken} (ST)  ;
		\draw (T) to node[above]{\nottaken} (SNT)  ;
		\draw (SNT) to[out=40,in=90] ++(.7,0) to[out=270,in=-40] node[below]{\nottaken} (SNT) ;
		\draw (SNT) to[bend right] node[right]{\taken} (NT)  ;
		\draw (NT) to[bend right] node[left]{\nottaken} (SNT)  ;
		\draw (NT) to node[below]{\taken} (ST)  ;
	\end{tikzpicture}%
	}
	\caption{%
		2-bit flip-on-consecutive predictor.
	}
	\label{fig:2-bit-flip-consecutive}
\end{figure}

The very limited memory of the 1-bit predictor fails on the frequent pattern of
a loop branch that is taken most of the time and only once in a while used to
leave the loop. 
In this situation, it would be much better to treat the loop exit as an outlier
that does not change the prediction.
2-bit predictors can achieve that.
They use 2 bits per branch instruction
to encode one of four states of a finite automaton.
The prediction is then made based on this current state and 
after the branch has executed, the state is updated according to the transition
matrix of the automaton.

There are two variants of 2-bit predictors described in the literature, that
use slightly different automata and thus give slightly different results.
The first one is the 2-bit saturating counter (2-bit~sc) shown in
\wref{fig:2-bit-saturating-counter}, which is used by
\citet{brodal2005tradeoffs} and \citet{Biggar2008}.
The second one is 2-bit flip-on-consecutive predictor (2-bit~fc) described 
by \citet{Kaligosi2006branch}; see below.
As both predictors are sensible choices, we analyze both.
Moreover, it is interesting to see the difference between the two;
see \wref{fig:comparison-steady-state-branch-miss} for that, as well.

To derive the steady-state miss-rate function, we translate the
automaton shown in \wref{fig:2-bit-saturating-counter} to a Markov chain,
compute its steady-state distribution and from that the expected misprediction rate.
Details are given in \wref{app:steady-state-predictability}.
The resulting function is 
\begin{align}
\label{eq:steady-state-predictability-2bitsc}
		\steadystatemissratetwobitsc(p)
	&\wwrel=
		\frac{p(1-p)}{1-2p(1-p)}  \;.
\end{align}

\subsubsection*{2-bit Flip-On-Consecutive Predictor.}

The second 2-bit variant flips its prediction only after two consecutive
mispredictions and we thus call it ``2-bit flip-on-consecutive predictor'' 
(see \wref{fig:2-bit-flip-consecutive}).
It is analyzed in the very same manner as 2-bit sc,
details are again given in \wref{app:steady-state-predictability}, where we find
\begin{align}
\label{eq:steady-state-predictability-2bitfc}
		\steadystatemissratetwobitfc(p)
	&\wwrel=
		\frac{2p^2(1-p)^2 + p(1-p)}{1-p(1-p)} \;.
\end{align}

\subsection{Results}

\begin{figure*}
	\def\adiv{\left\lfloor\!\frac{a-b}3\!\right\rfloor}
	\def\adivf{\left\lfloor\!\frac{a-b}4\!\right\rfloor}
	\resizebox{\linewidth}!{
	\parbox{\linewidth}{
	\begin{align*}
			\int_0^1 \!\!\! \frac{x^a(1-x)^b}{1-x(1-x)} \,dx
		&\wwrel=
			-\sum_{i=0}^{b-1} \BetaFun(a-i,b-i)
					\bin+ \sum_{i=1}^{\adiv}
					(-1)^{i-1} \bigl(
						\tfrac1{(a-b)-3i+2} + \tfrac1{(a-b)-3i+1}
					\bigr)
		\bin{+} \rho_1(a-b).
		\qquad\qquad (a\ge b)
	\\
				\int_0^1 \!\!\! \frac{x^a(1-x)^b}{\frac12-x(1-x)} \,dx
		&\wwrel=
			-\sum_{i=0}^{b-1} 2^{-i} \BetaFun(a-i,b-i)
				\bin+ 2^{-b} \sum_{i=1}^{\adivf}
					\bigl(-\tfrac14\bigr)^{i-1} \bigl(
						\tfrac1{(a-b)-4i+3} 
						+ \tfrac1{(a-b)-4i+2} + \tfrac{1/2}{(a-b)-4i+1}
					\bigr)
		\bin{+} 2^{-b}\rho_2(a-b).
	\end{align*}
	}
	}
	\footnotesize
	\def\adiv{\left\lfloor\!\frac{d}3\!\right\rfloor}
	\def\adivf{\left\lfloor\!\frac{d}4\!\right\rfloor}
	\vspace{-2ex}
	\begin{align*}
			\rho_1(d)
		&\wwrel{=}
				(-1)^{\adiv} \begin{cases}
					\frac{2\pi}{3\sqrt3} 	& \text{if } d \equiv 0 \pmod 3\\
					\frac{\pi}{3\sqrt3} 	& \text{if } d \equiv 1 \pmod 3\\				
					1-\frac{\pi}{3\sqrt3} 	& \text{if } d \equiv 2 \pmod 3\\
				\end{cases},
	&
			\rho_2(d)
		&\wwrel{=}
				\bigl(-\tfrac14\bigr)^{\adivf}\begin{cases}
					\pi 					& \text{if } d \equiv 0 \pmod 4\\
					\pi/2 					& \text{if } d \equiv 1 \pmod 4\\				
					1 						& \text{if } d \equiv 2 \pmod 4\\
					\frac32-\frac{\pi}{4} 	& \text{if } d \equiv 3 \pmod 4\\
				\end{cases}.
	\end{align*}
	\caption{%
		Explicit expressions for the integrals involved in 
		$\geoDirichletExp[1]ab$ and
		$\geoDirichletExp[2]ab$.
		The formulas are only valid for $a\ge b$, but since the
		integrals are symmetric, one can simply use $a' = \max\{a,b\}$ and $b'=\min\{a,b\}$.
		The proof consists in finding recurrences for the polynomial long division
		of the integrand, solving these recurrences and integrating them summand by summand.
		Details are given in \wref{app:distributions}.
	}
	\label{fig:geometric-beta-integrals}
\end{figure*}

Finally, we are in the position to put everything together.
As the involved constants become rather large,
we need to introduce some shorthand notation:
We write $\vect\tp = \vect t + 1 = (t_1+1,\ldots,t_s+1)$
and $\kp = k+1 = \tp_1 +\cdots+\tp_s$.
Moreover, $\rf xn$ denotes the $n$th rising factorial power of $x$ and by
$\geoDirichletExp[c]{a}{b}$ we denote the integral
\begin{align}
\label{eq:def-geo-dirichlet-exp}
		\geoDirichletExp[c]ab
	&\wrel=
		\frac1{\BetaFun(a,b)}
		\int_0^1 \frac{x^a(1-x)^b}{\frac1c - x(1-x)} \,dx \;.
\end{align}
where $\BetaFun(a,b)$ is the Beta function (see \wref{app:distributions}).
We will only need the cases $c=1$ and $c=2$, for which we have explicit,
but unwieldy expressions (see \wtpref{fig:geometric-beta-integrals}).

\begin{theorem}[Main Result]
\label{thm:main-results-BM} ~\\
	Let 
	$\E[\branchmisses_{\!n}^{\mathrm{CQS}}]$ and
	$\E[\branchmisses_{\!n}^{\mathrm{YQS}}]$ be 
	the expected number of branch misses incurred when 
	sorting a random permutation of length $n$ with 
	classic resp.\ Yaroslavskiy's Quicksort 
	under generalized pivot sampling with parameter 
	$\vect t = \vect\tp-1\in\N^s$.
	Then
	\begin{align*}
			\E[\branchmisses_{\!n}^{\mathrm{CQS}}]
		&\wwrel{\sim}
			\frac{a_{\mathrm{CQS}}}{\discreteEntropy[]}\, n\ln n \text{ and }
\\
			\E[\branchmisses_{\!n}^{\mathrm{YQS}}]
		&\wwrel{\sim}
			\frac{a_{\mathrm{YQS}}}{\discreteEntropy[]}\, n\ln n
			\,,
	\end{align*}
	with $\discreteEntropy[]=\discreteEntropy$ given in \wref{eq:discrete-entropy}
	and the constants
	\vspace{-2ex}
	\begin{align*}
			a_{\mathrm{CQS}}
		&\wrel{=}
			g_{\tp_1,\tp_2} \text{ and }
	\\
			a_{\mathrm{YQS}}
		&\wrel{=}
				\bigl(
					\tfrac{\tp_1}{\kp }		g_{\tp_1+1,\tp_2+\tp_3} + 
					\tfrac{\tp_2}{\kp }		g_{\tp_1,\tp_2+\tp_3+1}
				\bigr)
	\\* &\wwrel{\ppe}{}
				\bin+
				\bigl(
					\tfrac{\tp_1 \tp_2}{\rf\kp2}		g_{\tp_2+1,\tp_3} 
					+\tfrac{\tp_1 \tp_3}{\rf\kp2}		g_{\tp_2,\tp_3+1}
	\\*[-1ex] &\wwrel{\ppe}{}
					\qquad{}+\tfrac{\rf{\tp_2}2}{\rf\kp2}		g_{\tp_2+2,\tp_3} 
					+\tfrac{\tp_2 \tp_3}{\rf\kp2}		g_{\tp_2+1,\tp_3+1} 
				\bigr)
	\\	&\wwrel{\ppe}{}
				\bin+
				\bigl(
					\tfrac{\tp_3}{\kp }		g_{\tp_1+\tp_2,\tp_3+1}
				\bigr)
	\\* &\wwrel{\ppe}{}
				\bin+
				\bigl(
					\tfrac{\tp_1 \tp_3}{\rf\kp2}		g_{\tp_1+1,\tp_2} + 
					\tfrac{\tp_2 \tp_3}{\rf\kp2}		g_{\tp_1,\tp_2+1} 
				\bigr)
	\end{align*}
	where $g_{x,y}$ depends on the prediction scheme:
	\begin{enumerate}[label=(\roman*)]
		\item \mbox{\makeboxlike[l]{2-bit sc: }{1-bit:}}
		$
				g_{x,y}
			=
				2xy/\rf{(x+y)}2
		$,
		\item \mbox{\makeboxlike[l]{2-bit sc: }{2-bit sc:}}
		$
				g_{x,y}
			=
				\frac12 \geoDirichletExp[2]{x}{y}
		$,
		\item \mbox{\makeboxlike[l]{2-bit sc: }{2-bit fc:}}
		$
				g_{x,y}
			=
				\frac{2xy}{\rf{(x+y)}{2}}\geoDirichletExp[1]{x+1}{y+1}
				+\geoDirichletExp[1]{x}{y}
		$.
	\end{enumerate}
	In the limit $k\to\infty$ s.\,t.\ 
	$\vect t / k \to \vect \tau \in [0,1]^s$, we find
	$\E[\branchmisses_{\!n}] \sim \frac{a^*}{\contentropy[]}\, n\ln n$
	for $\contentropy[]$ defined in \wref{eq:limit-entropy} and an
	algorithm-dependent constant $a^*$\!:
	\begin{align*}
			a^*_{\mathrm{CQS}}
		&\wwrel=
			\steadystatemissrate{}(\tau_1) \text{ and}
	\\
			a^*_{\mathrm{YQS}}
		&\wwrel= %\phantom{{}\bin+{}}
			  (\tau_1+\tau_2) \cdot \steadystatemissrate{}(\tau_2+\tau_3)
	\\*	&\wwrel{\ppe}{}
			\bin+ (\tau_1+\tau_2)(\tau_2+\tau_3) \cdot \steadystatemissrate{}(\tau_2)
	\\*	&\wwrel{\ppe}{}
			\bin+ \tau_3 \cdot\steadystatemissrate{}(\tau_1+\tau_2)
	\\*	&\wwrel{\ppe}{}
			\bin+ \tau_3 (\tau_1+\tau_2) \cdot \steadystatemissrate{}(\tau_1) ,
	\end{align*}
	where $\steadystatemissrate{}$ is the steady-state 
	miss-rate function of the prediction scheme,
	see 
	\wref{eq:steady-state-predictability-1bit},
	\wref{eq:steady-state-predictability-2bitsc} 
	resp.~\wref{eq:steady-state-predictability-2bitfc}.
\end{theorem}

\begin{proof}
We plug the different steady-state miss rate functions into
\wref{eq:toll-CQS-generic} resp.\ \wref{eq:toll-YQS-generic}
and apply \wref{thm:leading-term-expectation-hennequin}.
What remains is to compute the leading term constants, \ie, 
$\E[\ui {C_n}i \steadystatemissrate{}(\btprob{i})]$ for all comparison locations.
\wpref{tab:cmps-execution-frequencies} gives the expected execution 
frequencies $\ui{C_n}i$ conditional on~$\vect D$ 
and the branch taken probabilities are 
$\btprob{c} = D_1$ for CQS and as given in
\wref{eq:branch-taken-prob-YQS} for YQS.

Using the properties of the Dirichlet distribution collected in 
\wref{app:distributions}, we can rewrite the involved expectations/integrals
until we can either evaluate them explicitly (as is the case for 1-bit)
or express them in terms of $\geoDirichletExp[c]ab$.
Full detail computations are given in \wref{app:computations}.

For the $k\to\infty$ part, one might compute the limit of the above terms;
however, there is a simpler direct argument:
For $k\to\infty$ s.\,t.\ $\vect t/k \to \vect \tau$ the $\dirichlet(\vect t+1)$ distribution
degenerates to a deterministic vector, \ie, $\vect D \to \vect\tau$ in probability.
By the \textsl{continuous mapping theorem}, we also have the limit (in probability)
$\steadystatemissrate{}(D_1) \to \steadystatemissrate{}(\tau_1)$ and thus 
$\E[\steadystatemissrate{}(D_1)] \to \steadystatemissrate{}(\tau_1)$.
\end{proof}

%% file: branch-misses-discussion.tex
\section{Discussion}
\label{sec:discussion}

\begin{table*}
	\def\fr{\sfrac}
	\def\ut{\tiny{\textsl{(large, but explicit term)}}}
	\small
	\plaincenter{%
	\begin{tabular}{lrllll}
	\toprule
	            &          & \multicolumn2{c}{Classic Quicksort (\generalClassictM)} 
	                       & \multicolumn2{c}{Yaroslavskiy (\generalYarostM)}\\
	\midrule               
	            & 1-bit    & $\fr{2}{3}$                           
	                          & $= 0.\overline6$ 
	                       & $\fr{101}{50}$                      
	                          & $= 0.67\overline3$ \\
	no sampling & 2-bit sc & $\fr{\pi}{2}-1$
	                          & $\approx 0.57080$                       
	                       & $\fr{31\pi}{40}-\fr{37}{20}$
	                          & $\approx 0.58473$\\
	            & 2-bit fc & $\fr{4\pi}{\sqrt3}-\fr{20}{3}$ 
	                          & $\approx 0.58853$ 
	                       & $\fr{49\pi}{5\sqrt{3}} - \fr{1288}{75}$
	                          & $ \approx 0.60190$       \\
	\midrule
	$k=5$, 
	            & 1-bit    & $\fr{180}{259}$
	                          & $\approx 0.69498$
	                       & $\fr{274}{399}$
	                          & $\approx 0.68671$  \\
	$\vect{t}_\mathrm{CQS}=(2,2)$,            
	            & 2-bit sc & $\fr{225\pi}{74}-\fr{330}{37}$
	                          & $\approx 0.63322$
	                       & $\fr{785\pi}{532} - \fr{535}{133}$
	                          & $\approx 0.61306$
	                       \\
	$\vect{t}_\mathrm{YQS}=(1,1,1)$
	            & 2-bit fc & $\fr{1200\pi\sqrt{3}}{37} - \fr{45\.540}{259}$
	                          & $\approx 0.64766$                
	                       & $\fr{1280\pi\sqrt{3}}{133} - \fr{20\.644}{399}$
	                          & $\approx 0.62899$ \\
	\midrule
	$k=5$, 
	            & 1-bit    & $\fr{600}{959}$
	                          & $\approx 0.62565$
	                       & $\fr{4070}{6419}$
	                          & $\approx0.63406$ \\
	$\vect{t}_\mathrm{CQS}=(4,0)$,
	            & 2-bit sc & $\fr{420}{137} - \fr{225\pi}{274}$
	                          & $\approx 0.48592$
	                       & $\fr{3135}{917} - \fr{3405\pi}{3668}$
	                          & $\approx 0.50242$ \\
	$\vect{t}_\mathrm{YQS}=(0,3,0)$             
	            & 2-bit fc & $\fr{23\.340}{959}-\fr{600\pi\sqrt{3}}{137}$
	                          & $\approx 0.50691$                     
	                       &  $\fr{335\.500}{6419} - \fr{8720\pi\sqrt{3}}{917}$
	                          & $\approx 0.52299$                      \\
	\midrule
	$k\to\infty$,
	            & 1-bit    & $\fr{1}{2\ln 2}$
	                          & $\approx 0.72135$
	                       & $\fr{7}{9\ln 3}$
	                          & $\approx 0.70796$  \\
	$\vect{\tau}_\mathrm{CQS}=(\fr12,\fr12)$, 
	            & 2-bit sc & $\fr{1}{2\ln 2}$
	            	          & $\approx 0.72135$
	                       & $\fr{11}{15 \ln 3}$
	                          & $\approx0.66751$ \\ 
	$\vect{\tau}_\mathrm{YQS}=(\fr13,\fr13,\fr13)$      
	            & 2-bit fc & $\fr{1}{2\ln 2}$
	            	          & $\approx 0.72135$
	                       & $\fr{47}{63\ln 3}$
	                          & $\approx0.67907$  \\
	\midrule
	$k\to\infty$,             
	            & 1-bit    & \ut
	                         &$\approx 0.55370$
	                       & \ut
	                         & $\approx 0.55987$             \\
	$\vect{\tau}_\mathrm{CQS}=(\fr1{10},\fr9{10})$
	            & 2-bit sc & \ut
	                         & $\approx 0.33762$
	                       & \ut
	                         & $\approx 0.34509$                          \\
	$\vect{\tau}_\mathrm{YQS}=(\fr1{10},\fr8{10},\fr1{10})$ 
	            & 2-bit fc & \ut 
	                         & $\approx 0.35900$
	                       & \ut 
	                         & $\approx 0.36746$             \\
	\bottomrule
	\end{tabular}%
	}
	\caption{%
		Coefficient of the leading term of BMs in 
		\generalClassictM and \generalYarostM
		for various combinations of branch prediction schemes
		and sampling parameters.
	}
	\label{tab:cqs-vs-yaroslavskiy}
\end{table*}

\begin{figure}
	\plaincenter{\includegraphics[width=0.66\linewidth]{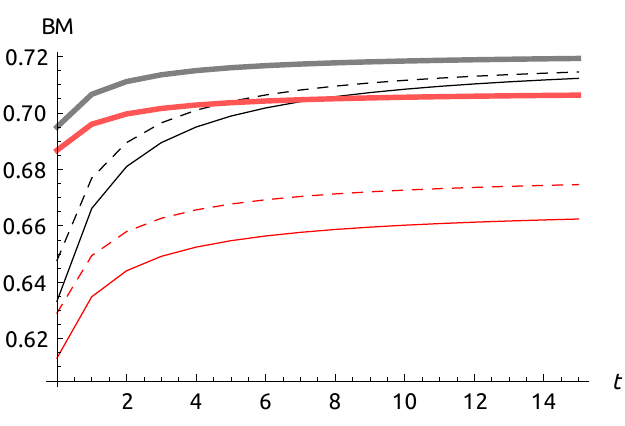}}
	\caption{%
		Branch mispredictions, as a function of~$t$, 
		in CQS (black) and YQS (red) with 1-bit branch prediction (fat), 
		2-bit saturating counter (thin solid) and 2-bit flip-consecutive (dashed)
		using symmetric sampling:
		$\vect{t}_\text{CQS}=(3t+2,3t+2)$ and
		$\vect{t}_\text{YQS}=(2t+1,2t+1,2t+1)$
	}
	\label{fig:BM-cqs-vs-yqs-sym}
\end{figure}

\begin{figure}
	\plaincenter{\includegraphics[width=0.66\linewidth]{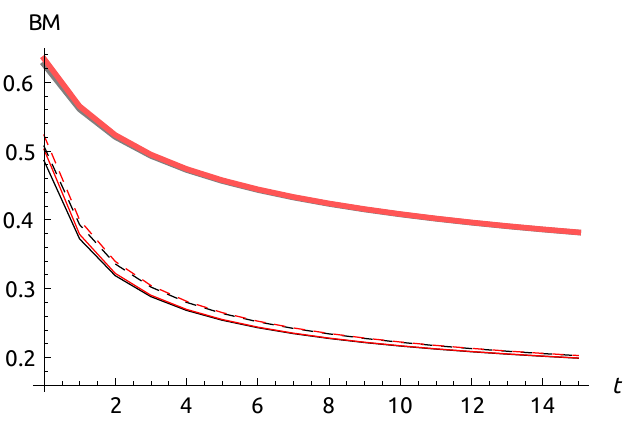}}
	\caption{%
		Branch mispredictions, as a function of~$t$, 
		in CQS (black) 
		and YQS (red) with 1-bit (fat), 
		2-bit sc (thin solid) and 2-bit fc (dashed) predictors,
		using extremely skewed sampling:
		$\vect{t}_\text{CQS}=(0,6t+4)$ and $\vect{t}_\text{YQS}=(0,6t+3,0)$
	}
	\label{fig:BM-cqs-vs-yqs-skew}
\end{figure}

\noindent
\wpref{tab:cqs-vs-yaroslavskiy} summarizes the leading factor (the
constant in front of $n\ln n$) in the total expected number
of branch mispredictions for both CQS and YQS 
under the various branch prediction schemes
and different pivot sampling strategies.

In practice, classic Quicksort implementations typically use median-of-3 
sampling, while in Oracle's YQS from Java~7 
the chosen pivots are the second
and the fourth in a sample of 5 (tertiles-of-5). 
With 1-bit prediction, this results in approximately
$0.6857\. n\ln n$ vs.\ $0.6867\. n\ln n$ BMs in the asymptotic average;
for the other branch prediction strategies the difference is similar.
It is very unlikely that the substantial differences in running times between
CQS and YQS are caused by this tiny difference in the number of branch misses.

\subsection{BM-Optimal Sampling}
\wref{fig:BM-cqs-vs-yqs-sym} shows the leading factor of BMs as a function of $t$,
where pivots are chosen equidistantly from samples of size $k=6t+5$, \ie,
in CQS we use the median as pivot, in YQS the tertiles.
Notice that, contrary to many other performance measures, 
sampling can be harmful with respect to branch mispredictions.
In particular, notice that with symmetric sampling (\ie, median-of-$(2t+1)$
for CQS, tertiles-of-$(3t+2)$ for YQS) the expected number of BMs 
\emph{increases} (and approaches a limit) as the size of the sample grows. 

We can weaken this undesirable effect by choosing \emph{skewed} pivots.
\wref{fig:BM-cqs-vs-yqs-skew} shows the same sample sizes as
\wref{fig:BM-cqs-vs-yqs-sym}, with the $\vect t$ that gives the minimal
number of BM for this sample size, which is to choose the extreme order statistics.
Not surprisingly, CQS and YQS 
behave almost identically when we use such highly skewed sampling parameters.
It is worth mentioning, though, that for \generalYarostM,
$\vect{t}=(0,6t+3,0)$ is better than
$\vect{t}=(6t+3,0,0)$ or $\vect{t}=(0,0,6t+3)$, as far as BMs
are concerned. Of course, such skewed $\vect{t}$ seriously penalize other
performance measures such as comparisons, cache misses, etc.\ as
unbalanced partitions become more likely.

In the limiting situation of very large samples, 
\ie, for $k\to\infty$ with $\vect t / k \to \vect \tau \in[0,1]^s$,
the leading coefficient of the expected number of BMs tends to~0 if we pick
the smallest or the largest element in the sample, \ie, 
$\vect{\tau}_\text{CQS}=(0,1)$ or $\vect{\tau}_\text{CQS}=(1,0)$ in CQS.
The same happens in YQS if we pick extremal pivots, \ie, if we take
$\vect{\tau}_\text{YQS}=(0,0,1)$ (the two smallest elements),
$\vect{\tau}_\text{YQS}=(0,1,0)$ (the smallest and the largest),
or $\vect{\tau}_\text{YQS}=(1,0,0)$ (the two largest elements).
This limiting behavior is independent of the branch prediction strategy,
but the trend doesn't show up unless the sample size is unrealistically large.
Due to its asymmetries, the convergence to the limit in YQS is not equal
for the three choices of $\vect{\tau}$.

\subsection{Overall Costs for CQS}

As the use of very skewed pivots severely runs against other important 
cost measures (including comparisons and cache misses), we need
to consider the \emph{combined} cost of several measures to determine
choices for~$\vect t$ with good practical performance. 
As a simplified model 
(which still exhibits the fundamental features of the problem) we shall 
consider the linear combination of low-level machine instructions 
(here: Bytecodes ($\bytecodes$)) and branch misses:
\[
Q \wwrel= \bytecodes \bin+ \xi \cdot \branchmisses.
\]  
The relative cost $\xi$ of one BM (in terms of BCs) depends on the machine.
As a mispredicted branch always entails a complete stall of the pipeline,
we need at least $\xi \ge L$ (the number of stages) to recover full speed.
Rollback of erroneously executed instructions might require another $L$ 
cycles of additional work, so $L \le \xi \le 2L$ seems a reasonable range.

As long as we are only interested in the main order term of $Q$, 
we only need the main order term of $\bytecodes$. 
For \generalClassictM with 
$\vect{t}=(t_1,t_2)$ we
have \citep{Martinez2001,Wild2012thesis}
\begin{multline*}
		\E[\toll[n]{\bytecodes}]
	\wwrel= 
\\		\left(6+18
        	\frac{(t_1+1)(t_2+1)}{(k+1)(k+2)} \right)\cdot n \bin+ \Oh(1) \,.
\end{multline*}
We are only interested in the coefficient $q_\xi(\vect{t})$ 
of the leading term of $\E[\toll[n]{Q}]$ which 
depends on $\xi$ and the sampling parameter $\vect{t}$.

\begin{figure}
	\plaincenter{%
	\includegraphics[width=0.66\linewidth]{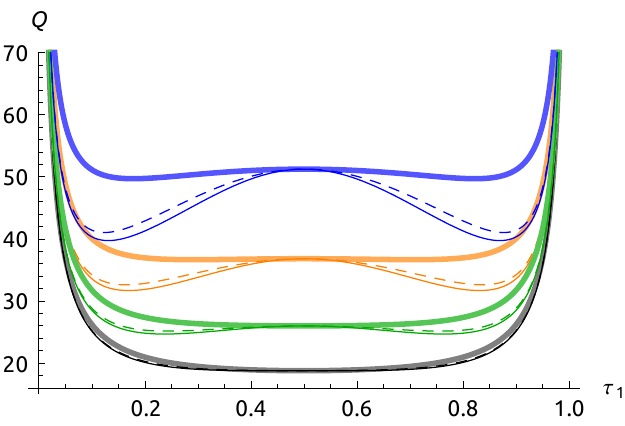}%
	}
	\caption{%
		The function $q_\xi(\vect\tau)$ for CQS with sample size
		$k\to\infty$ and $\vect{\tau}=(\tau_1,1-\tau_1)$, for 1-bit (fat lines), 
		2-bit sc (thin solid) and 2-bit fc (dashed). 
		For each branch prediction strategy, we give four plots 
		corresponding to $\xi=5$ (black), $\xi=15$ (green), $\xi=30$ (orange) and $\xi=50$ 
		(blue).}
	\label{fig:opt_q_lim_cqs}
\end{figure}

Again, we consider the limit $k\to\infty$ to reduce the parameter space 
and write
\[
		q_\xi(\vect{\tau}) 
	\wrel= 
		\lim_{\substack{k\to\infty:\\ \vect{t}/k\to\vect{\tau}} }
			q_\xi(\vect{t}).
\]
Figure~\ref{fig:opt_q_lim_cqs} shows the value of $q_\xi(\vect\tau)$ for several
values of $\xi$ and the three different
branch prediction strategies. 
For $\xi=5$, the three  prediction schemes yield almost identical results
as the weight of BMs is rather small. 
For larger values of $\xi$, the shape of the function $q_\xi(\vect\tau)$
changes and differences between the prediction schemes
become more pronounced.
The two 2-bit strategies do not differ
very significantly (2-bit sc is always slightly better), but 
both perform much better than 1-bit prediction.

\subsection{Total-Costs-Optimal Sampling in CQS}

It is natural to ask for choices of 
$\vect{t}$ (and implicitly $k=k(\vect{t})$) that
minimize $q_\xi(\vect t)$, for a given value of $\xi$.
In the limit of $k\to\infty$, we look for optimal $\vect{\tau^*}$.

It is well-known that if $\xi=0$ 
(the contribution of BMs to $q_\xi(\vect{t})$ is disregarded) 
for finite-size samples 
the best choice for the pivot is the median of the sample, thus
for any $k$,
\[
\vect{t}^\ast(0,k)=\bigr(\lfloor\tfrac{k-1}{2}\rfloor,
                        \lceil\tfrac{k-1}{2}\rceil\bigr).
\]
Here, by $\vect{t}^\ast=\vect{t}^\ast(\xi, k)$ 
we mean the sampling parameter that minimizes $q_\xi(\vect{t})$, for
fixed $\xi$ and sample size $k$. 
In the limit when $k\to\infty$, we thus have that
$\vect{\tau^*}=(\frac{1}{2},\frac{1}{2})$
as long as $\xi=0$ \cite{Martinez2001}. 

As can be seen in \wref{fig:opt_q_lim_cqs}, this changes for larger values of $\xi$:
There is a threshold~$\xi_c$ such that
for $\xi > \xi_c$ the optimal sampling parameter 
is
$\vect{\tau}^\ast=(\tau^\ast,1-\tau^\ast)$ with $\tau^\ast < 1/2$.
(By symmetry, $(1-\tau^\ast,\tau^\ast)$ is also optimal).
For $\xi>\xi_c$, $\tau^\ast$ is
the unique real solution in $[0,1/2)$ of
\[
	\tfrac{d}{d\tau} q_\xi\bigl((\tau,1-\tau)\bigr) \wwrel= 0.
\]
\wref{fig:opt_psi_cqs} shows the function $\tau^\ast=\tau^\ast(\xi)$
for our branch prediction strategies.
The critical threshold $\xi_c$ and
the shape of $\tau^\ast(\xi)$ differ depending on the
prediction scheme, but $\tau^\ast$ always
decreases steadily to 0 as $\xi\to\infty$, 
in accordance with our finding that the
number of BMs is minimized when ${\vect{t}}/{k}\to\vect{\tau}=(0,1)$. 
$\xi_c$ is the solution of the equation
\[
	\tfrac{d^2}{d\tau^2} q_\xi\bigl((\tau,1-\tau)\bigr) \Bigr|_{\tau=1/2} \wwrel= 0,
\]
because $\xi=\xi_c$ is the point where $\tau=1/2$ changes from
a local minimum to a local maximum.
The respective thresholds are
\begin{align*}
		\ui{\xi_c}{\text{1-bit}}
	&\wwrel= 
		\frac{3(7-6\ln 2)}{2\ln 2-1} 
		\wrel\approx 22.0644,
\\
		\ui{\xi_c}{\text{2-bit\,sc}}
	&\wwrel= 
		\frac{3(7-6\ln 2)}{4\ln 2-1} 
		\wrel\approx 4.8084,
\\
		\ui{\xi_c}{\text{2-bit\,fc}}
	&\wwrel= 
		\frac{9(7-6\ln 2)}{10\ln 2-3}
		\wrel\approx 6.5039.
\end{align*}

\begin{figure}
	\plaincenter{\includegraphics[width=0.66\linewidth]{%
		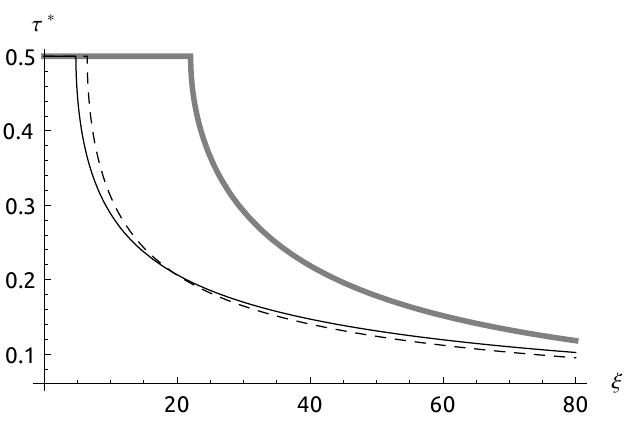%
	}}
	\caption{
		The function $\tau^\ast=\tau^\ast(\xi)$ for CQS with sample size
		$k\to\infty$ and $\vect{\tau}=(\tau,1-\tau)$, for 1-bit (fat line), 
		2-bit sc (solid thin) and 2-bit fc (dashed).
	}
	\label{fig:opt_psi_cqs}
\end{figure}

\subsection{Comparison with Experimental Data}

Kaligosi and Sanders~\cite{Kaligosi2006branch} report 
that running time was significantly improved by choosing skewed pivots
on a Pentium 4 Prescott processor.
They picked the $(\tau\cdot n)$-th of the array to partition 
(because the array contents where regular and known in advance); this same 
effect can be achieved via sampling with $k\to\infty$ and picking the 
$(\tau\cdot k)$-th within the sample at each stage~\cite{Martinez2001}.
They also did a (partial) analysis of BMs in the infinite-size sampling 
regime. From their experiments, 
they concluded that $\tau^\ast\approx 1/10$. 
If we assume that the Pentium 4 used 2-bit fc predictor
this corresponds to $\xi\approx 73$ (for 2-bit sc, we get $\xi\approx83$).
As $L=31$ for this processor~\citep{Biggar2008}, this means that $\xi$ is
a bit larger than expected\,---\,or that our simple model is not fully accurate here.
As the optimal $\tau^*$ is not much different for 
$\xi=30$ or $\xi=50$, see \wref{fig:opt_q_lim_cqs},
the model fits the observations satisfactorily.

Experiments conducted in other architectures (Opteron, Athlon, Sun) 
with shorter pipelines ($L\le20$) did not support the use of highly skewed
pivots~\citep{Biggar2008,Kaligosi2006branch}, 
and $\tau^\ast$ would be closer to 1/2.

In conclusion, despite we have proposed a very simplistic model
($\bytecodes+\xi \branchmisses$) of running time, it seems
to capture essential features of the problem (the most important
contribution missing is that due to memory hierarchies). 
The model has some explanatory power, even at a quantitative level, for
the phenomena observed in experiments.
As processor manufacturers often do not
provide information or just very sketchy descriptions on their architectures,
we can only ``guess'' the used branch prediction strategy and the value 
of $\xi$. 

\todoin[disable]{%
Put this back in for a full version
	
	For finite sizes, a closed form for $\vect{t}^\ast$ is not easy to obtain, but 
	specific values can be obtained rather easily. For brevity, let us 
	write $\vect{t}=(t,k-1-t)$. It is not 
	difficult to prove that when $\xi$ is small then $t^\ast=\lfloor(k-1)/2\rfloor$. As 
	$\xi$ grows $t^\ast$ decreases, and from some point on $t^\ast=0$. 
	The precise locations where $t^\ast$ changes from some value $r$ to $r-1$ 
	can be easily computed. Of course, these locations differ according to the
	branch prediction strategy, but in all three cases the qualitative behavior
	is very similar. Figure~\ref{fig:opt_t_cqs} shows $t^\ast$ as a function of 
	$\xi$ for the two branch prediction strategies (2-bit sc is not shown; 
	it is almost idential to 2-bit fc) and two values of $k$. The
	plots have been normalized to be on the same scale, they actually 
	show $(t^\ast+1)/(k+1)$.
	
	\begin{figure}
	\plaincenter{\includegraphics[width=0.66\linewidth]{plots/opt_t_cqs}}
	\caption{Normalized $t^\ast$ ($(t^\ast+1)/(k+1)$) in CQS for 1-bit (black) 
	and 2-bit fc (blue), when $k=5$ (solid) and $k=11$ (dashed). The functions 
	$\tau^\ast$ are also plotted for reference (black dotted, blue dotted).}
	\label{fig:opt_t_cqs}
	\end{figure}
	
	The computation of $\vect{\tau}^\ast$ (for infinite size samples) 
	and $\vect{t}^\ast$ (for finite size samples) in YQS goes along the same lines
	as we have shown for CQS. The trade-offs between instructions (BC) and 
	branch misses, and the inherent asymmetries of Yaroslavskiy's partitioning 
	algorithm also show up in non-symmetric optimal sampling parameters. 
	A detailed study, including plots, numerical optimization, etc.
	will appear in a full version of this paper.

	For instance, even when $\xi=0$, 
	$\vect{\tau}^\ast=(0.206772, 0.348562, 0.444666)$~\cite{NebelWild2014};
	with this $\vect{\tau}^\ast$ the expected 
	number of bytecodes is $16.383\ldots n\ln n$ cf.\ the 
	$16.991\ldots n\ln n$ bytecodes that we get with the symmetric 
	sampling parameter 
	$\vect{\tau}=\left(\frac13, \frac13, \frac13\right)$.
}

\todoin[disable]{%
	We have the following results:
	\begin{itemize}
		\item 
		Comparison with results from \cite{Kaligosi2006branch}:
		Beware: \citet{Kaligosi2006branch} call static prediction an optimal, oracle-based 
		predictor, \emph{not} a compile-time static prediction.
		
		1bit: same results
		
		2bit (fc): they do ``worst case'', their random pivot result is only an upper bound,
		the results for $\alpha$-skewed pivots are exactly the same
		
		\item
		Comparison of CQS and YQS for some small realistic $\vect t$.
		Result: Without sampling YQS slighty worse, with symmetric sampling slightly better:
		We compare $\vect t = (3r-1,3r-1)$ CQS to $\vect t=(2r-1,2r-1,2r-1)$ YQS, \ie, 
		for both we have sample size $k=6r-1$.
		CQS/YQS is between 1 and 2\,\% for 1-bit and between 2 and 6 \% for 2-bit fc
		
		\item
		optimal $\vect t$ for linear combination costs
	\end{itemize}
}

%% file: branch-misses-notation.tex
\section{Index of Used Notation}
\label{app:notations}
\def\mydots{\xleaders\hbox to.75em{\hfill.\hfill}\hfill}

\newlength\tmpLenNotations
\newenvironment{notations}[1][8em]{%
	\small
	\newcommand\notationentry[1]{%
		\settowidth\tmpLenNotations{##1}%
		\ifthenelse{\lengthtest{\tmpLenNotations > \labelwidth}}{%
			\parbox[b]{\labelwidth}{%
				\makebox[0pt][l]{##1}\\%
			}%
		}{%
			\mbox{##1}%
		}%
		\mydots\relax%
	}%
	\begin{list}{}{%
		\setlength\labelsep{0em}%
		\setlength\labelwidth{#1}%
		\setlength\leftmargin{\labelwidth+\labelsep+1em}%
		\renewcommand\makelabel{\notationentry}%
	}
	\newcommand\notation[1]{\item[##1]}
	\raggedright
}{%
	\end{list}
}

In this section, we collect the notations used in this paper.
(Some might be seen as ``standard'', but we think
including them here hurts less than a potential misunderstanding caused by
omitting them.)

\subsubsection*{Generic Mathematical Notation}
\begin{notations}
\notation{$\ln n$}
	natural logarithm.
\notation{$\vect x$}
	to emphasize that $\vect x$ is a vector, it is written in \textbf{bold};\\
	components of the vector are not written in bold: $\vect x = (x_1,\ldots,x_d)$.\\
	operations on vectors are understood elementwise, \eg,
	$\vect x+1 = (x_1+1,\ldots,x_d+1)$
\notation{$X$}
	to emphasize that $X$ is a random variable it is Capitalized.
\notation{$\harm{n}$}
	$n$th harmonic number; $\harm n = \sum_{i=1}^n 1/i$.
\notation{$\dirichlet(\vect \alpha)$}
	Dirichlet distributed random variable with parameter
	$\vect \alpha \in \R_{>0}^d$; 
	see \wref{app:distributions}.
\notation{$\uniform(a,b)$}
	uniformly in $(a,b)\subset\R$ distributed random variable. 
\notation{$\gammadist(k,\theta)$}
	Gamma distributed random variable with 
	shape parameter $k\in\R_{>0}$ and scale parameter $\theta\in\R_{>0}$.
\notation{$\gammadist(k)$}
	$\gammadist(k,1)$ distributed random variable.
\notation{$\BetaFun(\alpha_1,\ldots,\alpha_d)$}
	$d$-dimensional Beta function; 
	given in \wildpageref[equation]{eq:def-beta-function}{\eqref}.%
\notation{{$\E[X]$}}
	expected value of $X$.
\notation{{$\E[X\given Y]$}}
	the conditional expectation of $X$ given $Y$.
\notation{$\Prob(E)$, $\Prob(X=x)$}
	probability of an event $E$ resp.\ probability for random variable $X$ to
	attain value $x$.
\notation{$X\eqdist Y$}
	equality in distribution; $X$ and $Y$ have the same distribution.
\notation{$X_{(i)}$}
	$i$th order statistic of a set of random variables $X_1,\ldots,X_n$,\\
	\ie, the $i$th smallest element of $X_1,\ldots,X_n$. 
\notation{$a^{\underline b}$, $a^{\overline b}$}
	factorial powers notation of \citep{ConcreteMathematics}; 
	``$a$ to the $b$ falling resp.\ rising''.
\notation{$\dirichletExpectation{f(\vect X)}{\vect\alpha}$}
	expectation of $f(\vect X)$ over $\dirichlet(\vect \alpha)$ distributed $\vect X$; 
	formally for $\vect\alpha \in \R_{>0}^d$\\
	$ \displaystyle
			\dirichletExpectation{f(\vect X)}{\vect\alpha} 
		= 
			\int_{\Delta_d} f(\vect x) 
				\frac{x_1^{\alpha_1-1}\cdots x_d^{\alpha_d-1}}{\BetaFun(\vect\alpha)}
			\mu(d\vect x)
	$,\\
	where $\Delta_d$ is the standard $(d-1)$-dimensional simplex, see \wref{eq:def-delta-d}.
\notation{{$\geoDirichletExp[c]ab$}}
	see \wildtpageref[equation]{eq:def-geo-dirichlet-exp}{\eqref}{}%
		; $
				\geoDirichletExp[c]ab
			= 
				\dirichletExpectation{\frac{X_1 X_2}{\frac1c-X_1 X_2}}{a,b}
		$
\notation{{$\discreteEntropy[] = \discreteEntropy[\vect t]$}}
	discrete entropy; defined in
	\wildpageref[equation]{eq:discrete-entropy}{\eqref}.
\notation{{$\contentropy[] = \contentropy[\vect p]$}}
	continuous (Shannon) entropy with base $e$; given in
	\wildpageref[equation]{eq:limit-entropy}{\eqref}.
\end{notations}

\subsubsection*{Input to the Algorithm}
\begin{notations}
\notation{$n$}
	length of the input array, \ie, the input size.
\notation{\arrayA}
	input array containing the items $\arrayA[1],\ldots,\arrayA[n]$ to be
	sorted; initially, $\arrayA[i] = U_i$.
\notation{$U_i$}
	$i$th element of the input, \ie, initially $\arrayA[i] = U_i$.\\
	We assume $U_1,\ldots,U_n$ are i.\,i.\,d.\ $\uniform(0,1)$ distributed.
\end{notations}

\subsubsection*{Notation Specific to the Algorithms}
\begin{notations}
\notation{$s$}
	number of subproblems, \ie, $s=2$ for classic Quicksort
	and $s=3$ for Yaroslavskiy's Quicksort;
	$s-1$ is thus the number of pivots in one partitioning step.
\notation{CQS, YQS}
	CQS is the abbreviation for \textbf classic (single-pivot) 
	\textbf Quick\textbf sort 
	using Sedgewick's partitioning given in \wref{alg:partition-sedgewick};
	YQS likewise stands for \textbf Yaroslavskiy's (dual-pivot) 
	\textbf Quick\textbf sort using the partitioning given in 
	\wref{alg:partition-yaroslavskiy}.
\notation{$\vect t \in \N^{s}$}
	pivot sampling parameter.
\notation{$k=k(\vect t)$}
	sample size; defined in terms of $\vect t$ as 
	$k(\vect t) = t_1+t_2+\cdots+t_{s}+ (s-1) = \|\vect t\|_1 + \dim(\vect t)-1$.
\notation{\isthreshold}
	Insertionsort threshold; for $n\le\isthreshold$, Quicksort recursion is
	truncated and we sort the subarray by Insertionsort.
\notation{$\generalClassictM$, $\generalYarostM$}
	$\generalClassictM$ is the 
	abbreviation for generalized classic Quicksort,
	where the pivot is chosen by generalized pivot sampling with parameter $\vect t$
	and where we switch to Insertionsort for subproblems of size at most
	\isthreshold;
	similarly $\generalYarostM$ for Yaroslavskiy's dual-pivot partitioning.
\notation{$\vect V \in \N^k$}
	(random) sample for choosing pivots in the first partitioning step.
\notation{$P$, $Q$}
	(random) values of chosen pivots in the first partitioning step;
	for classic Quicksort, only $P$.
\notation{{$\vect D\in[0,1]^{s+1}$}}
	(random) spacings of the unit interval $(0,1)$ induced by the pivot(s), \ie, 
	for CQS we have $\vect D = (P, 1-P)$ and
	for YQS $\vect D = (P,Q-P,1-Q)$;
	in both cases, we have $\vect D \eqdist \dirichlet(\vect t + 1)$.
\notation{small element}
	element $U$ is small if $U<P$.
\notation{medium element}
	(only for YQS) element $U$ is medium if $P<U<Q$.
\notation{large element}
	for CQS: element $U$ is large if $P<U$;
	for YQS: if $Q < U$.
\notation{ordinary element}
	the $n-k$ array elements that have not been part of the sample.
\notation{$k$, $g$, $\ell$}
	index variables used in the partitioning methods, see
	\wref{alg:partition-sedgewick} and \wpref{alg:partition-yaroslavskiy};
	$\ell$ only appears in \wref{alg:partition-yaroslavskiy}.
\end{notations}

\subsubsection*{Notation for Analysis of Branch Misses}
\begin{notations}
\notation{BM}
	shorthand for \textbf branch \textbf miss
\notation{$\bmonebit$, $\bmtwobitsc$, $\bmtwobitfc$}
	branch prediction schemes analyzed in this paper; see \wref{sec:branch-prediction-schemes}.
\notation{$\ui C{c1}, \ui C{c2}$}
	key comparison locations in CQS: $\ui C{c1}$ corresponds to the comparison in 
	\wref{lin:classic-comp-1}, $\ui C{c2}$ to \wref{lin:classic-comp-2} of
	\wref{alg:partition-sedgewick}.
\notation{$\ui Cc$}
	the virtual combined comparison location consisting of $\ui C{c1}$ and $\ui C{c2}$.
\notation{$\ui C{y1},\ldots,\ui C{y4}$}
	key comparison locations in YQS: 
	$\ui C{y1}$ is in \wref{lin:yaroslavskiy-comp-1},
	$\ui C{y2}$ in \wref{lin:yaroslavskiy-comp-2},
	$\ui C{y3}$ in \wref{lin:yaroslavskiy-comp-3} and
	$\ui C{y4}$ is in \wref{lin:yaroslavskiy-comp-4} of
	\wref{alg:partition-yaroslavskiy}.
\notation{$\ui {C_n}i$}
	$i\in\{c1,c2,c,y1,y2,y3,y4\}$;
	expected execution frequency of comparison location $\ui Ci$ in the first 
	partitioning step.
\notation{$\btprob{i}$}
	$i\in\{c1,c2,c,y1,y2,y3,y4\}$;
	probability (conditional on $\vect D$) 
	for the branch at comparison location $\ui Ci$ to be taken;
	all executions of this branch are i.\,i.\,d.
\notation{$\bmprob{i}$}
	$i\in\{c1,c2,c,y1,y2,y3,y4\}$;
	probability (conditional on $\vect D$) for the branch at comparison location
	$\ui Ci$ to be mispredicted; depends on the prediction scheme.
\notation{$\steadystatemissrate{\mathit{bps}}$, $\steadystatemissrate{\mathit{bps}}(p)$}
	$\mathit{bps}\in\{\bmonebit,\bmtwobitsc,\bmtwobitfc\}$;\\
	steady-state miss-rate function for the branch prediction scheme $\mathit{bps}$;
	$\steadystatemissrate{\mathit{bps}}(p)$ is the probability for a BM 
	at a branch that is i.\,i.\,d.\ and taken with probability $p$\,---\,in the long run, \ie,
	when the predictor automaton has reached its steady state.
\notation{$\branchmisses_{\!n}$}
	(random) total number of branch misses to sort a random permutation of length $n$.
\notation{\toll{\circ}}
	(random) number of branch misses ($\circ=\branchmisses$), Bytecodes ($\circ=BC$) 
	resp.\ combined costs ($\circ=Q$) 
	of the first partitioning step on a random permutation of size~$n$;
	$\toll[n]{\circ}$ when we want to emphasize dependence on~$n$.
\notation{$a$}
	coefficient of the linear term of $\E[\toll[n]{\,}]$ 
	see \wpref{thm:leading-term-expectation-hennequin}.
\notation{$\vect\tau$}
	for the limiting case when $k\to \infty$, we assume that 
	$\vect t / k = (t_1/k , \ldots, t_s/k) \to \vect \tau$.
	This corresponds to selecting precise order statistics, s.\,t.\
	the relative size of subproblem $i$ becomes precisely $\tau_i$.
\notation{$a^*$}
	limit of $a$ when $k\to\infty$ and $\vect t/k \to \vect\tau$.
\notation{$\xi$}
	cost of one branch miss relative to one CPU cycle;
	used in $Q$.
\notation{$L$}
	length of the instruction pipeline, \ie, the number of stages each instruction is broken into.
\notation{$Q$}
	combined cost measure $Q=\bytecodes + \xi \cdot \branchmisses$, 
\end{notations}

%% file: branch-misses-steady-state-predictability.tex
\section{Steady-State Miss-Rate Functions}
\label{app:steady-state-predictability}

In this appendix, we give details on the computation of the steady-state miss-rate
for the two 2-bit predictor variants.
These functions have also been derived by \citet{Biggar2008} and \citet{Kaligosi2006branch}, 
respectively. 
We show the derivations here again for the reader's convenience.

\subsection{2-Bit Saturating Counter}

Consider a branch instruction, which is i.\,i.\,d.\ taken with probability $p$.
The saturating-counter automaton then corresponds to the following Markov chain with
states ordered as in \wpref{fig:2-bit-saturating-counter}:
\begin{align}
		\Pi \wrel= \Pi(p)
	&\wwrel=
		\left(
			\begin{array}{cccc}
			 p & 1-p & 0 & 0 \\
			 p & 0 & 1-p & 0 \\
			 0 & p & 0 & 1-p \\
			 0 & 0 & p & 1-p \\
			\end{array}
		\right)\;.
\end{align}
The stationary distribution is found to be
\begin{align}
		\pi(p)
	&\wwrel=
		\frac1{1-2p(1-p)} 
\\*	&\wwrel{\ppe}{}\cdot
		\left(
			p^3, p^2(1-p), p (1-p)^2, (1-p)^3
		\right).
\end{align}
The next branch execution is mispredicted iff we are in state 1 or 2 and the
branch was not taken or if we are in state 3 or 4, but the branch was taken.
Thus if we assume the predictor has reached its steady-state distribution, then 
we obtain the branch miss probability in dependence of the branch
probability $p$ as
\begin{align}
		\steadystatemissratetwobitsc(p)
	&\wwrel=
		\pi(p) \cdot (1-p,1-p,p,p)^T
\\	&\wwrel=
		\frac{p(1-p)}{1-2p(1-p)}  \;.
\end{align}

\subsection{2-Bit Flip-On-Consecutive}

The analysis is along the same lines as above.
The underlying Markov chain for the 2-bit fc predictor is shown in 
\wpref{fig:2-bit-flip-consecutive}
and has the transition matrix
\begin{align}
		\tilde\Pi \wrel= \tilde\Pi(p)
	&\wwrel=
		\left(
			\begin{array}{cccc}
			 p & 1-p & 0 & 0 \\
			 p & 0 & 0 & 1-p \\
			 p & 0 & 0 & 1-p \\
			 0 & 0 & p & 1-p \\
			\end{array}
		\right)\;.
\end{align}
with stationary distribution
\begin{align}
		\tilde\pi(p)
	&\wwrel=
		\frac1{1-p(1-p)} 
\\*	&\wwrel{\ppe}{}\cdot
		\left(
			p^2, p^2(1-p), p (1-p)^2, (1-p)^2
		\right) .
\end{align}
Again assuming this steady state has been reached, the branch miss
probability is
\begin{align}
		\steadystatemissratetwobitfc(p)
	&\wwrel=
		\tilde\pi(p) \cdot (1-p,1-p,p,p)^T
\\	&\wwrel=
		\frac{2p^2(1-p)^2 + p(1-p)}{1-p(1-p)}  \,.
\end{align}  

%% file: branch-misses-preliminaries.tex
\section{Properties of the Dirichlet Distribution}
\label{app:distributions}

We herein collect definitions and basic properties of the \textsl{Dirichlet distribution},
which plays a central r\^ole for the analyses in this paper.
We use the notation $x^{\overline n}$ and $x^{\underline n}$ of
\citet{ConcreteMathematics} for rising and falling factorial powers, respectively.

For $d\in\N$ let $\Delta_d$ be the standard $(d-1)$-dimensional simplex, \ie, 
\begin{align}
\label{eq:def-delta-d}
		\Delta_d
	&\wwrel\ce 
		\biggl\{
			(x_1,\ldots,x_d) 
			\wrel: 
			\forall i : x_i \ge 0 \; 
			\rel\wedge 
			\sum_{\mathclap{1\le i \le d}} x_i = 1
		\biggr\} \;.
\end{align}
Let $\alpha_1,\ldots,\alpha_d > 0$ be positive reals.
A random variable $\vect X \in \R^d$ is said to have the 
\emph{Dirichlet distribution} with \emph{shape parameter} 
$\vect\alpha \ce (\alpha_1,\ldots,\alpha_d)$\,---\,abbreviated as
$\vect X \eqdist \dirichlet(\vect\alpha)$\,---\,if it has a density given by
\begin{multline}
\label{eq:def-dirichlet-density}
		f_{\vect X}(x_1,\ldots,x_d)
	\wwrel\ce 
	\\
	\begin{cases}
			\frac1{\BetaFun(\vect\alpha)} \cdot
			x_1^{\alpha_1 - 1} \cdots x_d^{\alpha_d-1} ,
			& \text{if } \vect x \in \Delta_d \,; \\
			0 , & \text{otherwise} \..
		\end{cases}
\end{multline}
Here, $\BetaFun(\vect\alpha)$ is the \emph{$d$-dimensional Beta function}
defined as the following Lebesgue integral:
\begin{align}
\label{eq:def-beta-function}
		\BetaFun(\alpha_1,\ldots,\alpha_d)
	&\wwrel\ce
		\int_{\Delta_d} x_1^{\alpha_1 - 1} \cdots x_d^{\alpha_d-1} \; \mu(d \vect x)
	\;.
\end{align}
The integrand is exactly the density without the normalization constant
$\frac1{\BetaFun(\alpha)}$, hence $\int f_{\vect X} \,d\mu = 1$ as needed for
probability distributions.

The Beta function can be written in terms of the Gamma function 
$\Gamma(t) = \int_0^\infty x^{t-1} e^{-x} \,dx$
as
\begin{align}
\label{eq:beta-function-via-gamma}
		\BetaFun(\alpha_1,\ldots,\alpha_d)
	&\wwrel=
		\frac{\Gamma(\alpha_1) \cdots \Gamma(\alpha_d)}
		{\Gamma(\alpha_1+\cdots+\alpha_d)} \;.
\end{align}
(For integral parameters $\vect\alpha$, a simple inductive argument and partial
integration suffice to prove~\wref{eq:beta-function-via-gamma}.)

Note that $\dirichlet(1,\ldots,1)$ corresponds to the uniform distribution over
$\Delta_d$.
For integral parameters $\vect\alpha\in\N^d$, $\dirichlet(\vect\alpha)$ is the
distribution of the \emph{spacings} or \emph{consecutive differences} induced by
appropriate order statistics of i.\,i.\,d.\ uniformly in $(0,1)$ distributed
random variables: 
\begin{proposition}[{{\cite[Section\,6.4]{David2003}}}]
\label{pro:spacings-dirichlet-general-dimension}
~\\
	Let $\vect\alpha \in \N^d$ be a vector of positive integers and set $k \ce -1 +
	\sum_{i=1}^d \alpha_i$. Further let $V_1,\ldots,V_{k}$ be $k$ random variables
	i.\,i.\,d.\ uniformly in $(0,1)$ distributed.
	Denote by \smash{$V_{(1)}\le \cdots \le V_{(k)}$} their corresponding
	order statistics.
	We select some of the order statistics according to $\vect \alpha$: 
	For $j=1,\ldots,d-1$ define \smash{$W_j \ce V_{(p_j)}$}, where $p_j \ce
	\sum_{i=1}^j \alpha_i$. Additionally, we set $W_0 \ce 0$ and $W_d \ce 1$.
	
	Then, the \textit{consecutive distances} (or \textit{spacings}) $D_j \ce W_j -
	W_{j-1}$ for $j=1,\ldots,d$ induced by the selected
	order statistics $W_1,\ldots,W_{d-1}$ are Dirichlet
	distributed with parameter $\vect \alpha$: 
	\begin{align*}
			(D_1,\ldots,D_d) 
		&\wwrel\eqdist 
			\dirichlet(\alpha_1,\ldots,\alpha_d) \;.
	\end{align*}%
\qed\end{proposition}
\smallskip

In the computations of expected partitioning costs, mixed moments of Dirichlet
distributed variables show up, which can be dealt with using the following general
statement for $f\equiv1$.

\begin{lemma}[``Powers-to-Parameters'']
\label{lem:dirichlet-powers-to-parameters}
~\\
	Let $\vect X = (X_1,\ldots,X_d) \in \R^d$ be a $\dirichlet(\vect\alpha)$
	distributed random variable with parameter $\vect\alpha = (\alpha_1,\ldots,\alpha_d)$.
	Let further $\vect m = (m_1,\ldots,m_d) \in \Z^d$
	be an integer vector with $\vect m > -\vect\alpha$ (componentwise) and abbreviate
	the sums $A \ce \sum_{i=1}^d \alpha_i$ and $M \ce \sum_{i=1}^d m_i$. 
	Then we	have for an arbitrary (real-valued) function $f : \Delta_{d} \to \R$ the identity
	\begin{align*}
			\E\bigl[ X_1^{m_1} \cdots X_d^{m_d} \cdot f(\vect X) \bigr]
		&\wwrel=
			\frac{\alpha_1^{\overline{m_1}} \cdots \alpha_d^{\overline{m_d}}}
				{A^{\overline M}} 
			\cdot \E\bigl[ f(\vect{\tilde X}) \bigr]
			\,,
	\end{align*}
	where $\vect{\tilde X} = (\tilde X_1,\ldots,\tilde X_d)$ is
	$\dirichlet(\vect\alpha+\vect m)$ distributed.
\end{lemma}

\begin{proof}
Using $\frac{\Gamma(z+n)}{\Gamma(z)} = z^{\overline n}$ for all
$z\in\R_{>0}$ and $n \in \N$, we compute
\begin{align*}
		\E\bigl[ X_1^{m_1} \cdots X_d^{m_d} \cdot f(\vect X) \bigr]
	\mkern-150mu\\
	&\wwrel=
		\int_{\Delta_d} 
			x_1^{m_1} \cdots x_d^{m_d} f(\vect x) \cdot
			\frac{x_1^{\alpha_1-1} \cdots x_d^{\alpha_d-1}}{\BetaFun(\vect\alpha)}
		\; \mu(d\vect x)
	\\	&\wwrel=
		\frac{\BetaFun(\vect\alpha+\vect m)}{\BetaFun(\vect\alpha)}
		\cdot
		\int_{\Delta_d} 
		f(\vect x) \cdot
	\\*&\wwrel\ppe\qquad\qquad
		\frac{x_1^{\alpha_1+m_1-1} \cdots x_d^{\alpha_d+m_d-1}}
			{\BetaFun(\vect\alpha+\vect m)}
		\; \mu(d\vect x)
	\\	&\wwrel=
		\frac{\BetaFun(\alpha_1 + m_1,\ldots,\alpha_d + m_d)}
			{\BetaFun(\alpha_1,\ldots,\alpha_d)}
		\cdot \E\bigl[ f(\vect{\tilde X}) \bigr]
	\\	&\wwrel{\eqwithref{eq:beta-function-via-gamma}}
		\frac{ \alpha_1^{\overline{m_1}} \cdots \alpha_d^{\overline{m_d}} }
			{ A^{\overline M} } 
			\cdot \E\bigl[ f(\vect{\tilde X}) \bigr]
		\;.
\end{align*}
\end{proof}

\begin{theorem}
\label{thm:dirichlet-characterization-gamma}
~\\
{\sc(``Characterization via Gamma'' \citep[Thm~4.1]{Devroye1986})}\\[.5ex]
	Let $\vect X=(X_1,\ldots,X_d)$ be $\dirichlet(\vect\alpha)$ distributed with
	$\vect\alpha\in\R_{>0}^d$ and let $G_1,\ldots,G_d$ be $d$ independent
	\textsl{Gamma distributed} variables with parameters $\vect\alpha$, \ie,
	$G_i \eqdist \gammadist(\alpha_i)$.
	Further define $S = G_1+\cdots+G_d$. Then
	$$
		\vect X
		\wwrel{\eqdist}
		\biggl(
			\frac{G_1}S,\ldots,\frac{G_d}S
		\biggr) .
	$$
\qed\end{theorem}

\begin{lemma}[``Aggregation'']
\label{lem:dirichlet-aggregation}
~\\
	Let $\vect X = (X_1,\ldots,X_d) \in \R^d$ be a $\dirichlet(\vect\alpha)$
	distributed random variable with parameter $\vect\alpha = (\alpha_1,\ldots,\alpha_d)$.
	For two components $1\le i <j\le d$ define the ``aggregated vector''
	$
		\vect{\hat X} \wrel= (\hat X_1,\ldots,\hat X_{d-1})
	$
	with
	$$
		\hat X_l 
		\wwrel=
		\begin{cases}
			X_l, 		& \text{if } l < j \wedge l \ne i; \\
			X_i + X_j,	& \text{if } l = i; \\
			X_{l+1},	& \text{if } l \ge j,
		\end{cases}
	$$
	i.\,e., the $i$th and $j$th components have been added up.
	Then, $\vect{\hat X}$ is $\dirichlet(\vect{\hat\alpha})$ distributed with
	the aggregated parameter vector
	$\vect{\hat\alpha} = (\hat\alpha_1,\ldots,\hat\alpha_{d-1})$ where
	$$
		\hat \alpha_l 
		\wwrel=
		\begin{cases}
			\alpha_l, 		& \text{if } l < j \wedge l \ne i; \\
			\alpha_i + \alpha_j,	& \text{if } l = i; \\
			\alpha_{l+1},	& \text{if } l \ge j.
		\end{cases}
	$$
\end{lemma}
\begin{proof}
Using \wref{thm:dirichlet-characterization-gamma}, we have
$\vect X \eqdist \vect G / S$ for $\vect G = (G_1,\ldots,G_d)$ with
$G_i \eqdist \gammadist(\alpha_i)$ for $i=1,\ldots,d$ and $S=G_1+\cdots+G_d$.
Defining the aggregated Gamma vector 
$
	\vect{\hat G} \wrel= (\hat G_1,\ldots,\hat G_{d-1})
$
analogously to $\vect{\hat X}$ via
$$
	\hat G_l 
	\wwrel=
	\begin{cases}
		G_l, 		& \text{if } l < j \wedge l \ne i; \\
		G_i + G_j,	& \text{if } l = i; \\
		G_{l+1},	& \text{if } l \ge j,
	\end{cases}
$$
we get again by \wref{thm:dirichlet-characterization-gamma} that
$\vect{\hat X} = \vect{\hat G} / S$ has a $\dirichlet(\vect{\hat{\alpha}})$ distribution,
as claimed.
\end{proof}

\begin{lemma}[``Zoom'']
\label{lem:dirichlet-zoom}
	Let $\vect X = (X_1,\ldots,X_d) \in \R^d$ be a $\dirichlet(\vect\alpha)$
	distributed random variable with parameter $\vect\alpha = (\alpha_1,\ldots,\alpha_d)$.
	For two components $1\le i <j\le d$ define the ``zoom variable''
	$X^\circ = \frac{X_i}{X_i+X_j}$.
	Then, $(X^\circ,1-X^\circ)$ is $\dirichlet(\alpha_i,\alpha_j)$ distributed.
\end{lemma}
\begin{proof}
	As in the proof of \wref{lem:dirichlet-aggregation}, we use
	\wref{thm:dirichlet-characterization-gamma} to write
	$\vect X \eqdist \vect G / S$ for $\vect G = (G_1,\ldots,G_d)$ with
	$G_i \eqdist \gammadist(\alpha_i)$ for $i=1,\ldots,d$ and $S=G_1+\cdots+G_d$.
	Then, $X^\circ \eqdist G_i / (G_i+G_j)$ and 
	by applying \wref{thm:dirichlet-characterization-gamma} once more,
	$X^\circ \eqdist \dirichlet(\alpha_i,\alpha_j)$.
\end{proof}

\begin{lemma}[``Geometric Beta integrals'']
\label{lem:geometric-betas}
~\\
	\def\adiv{\left\lfloor\!\frac{a-b}3\!\right\rfloor}
	\def\adivf{\left\lfloor\!\frac{a-b}4\!\right\rfloor}
	Let $a\ge b > 0$ be two integer constants.
	Then with $c=1$ or $c=2$, the integral
	$$
%			\geoBeta[c]{a}{b}
%		\wwrel\ce
			\int_0^1 \frac{x^a(1-x)^b}{\frac1c-x(1-x)} \,dx
	$$
	has the value given in \wpref{fig:geometric-beta-integrals}.
\end{lemma}

We require $a\ge b$ just for convenience of notation; 
the integrals are symmetric in $a$ and $b$.

\begin{proof}
	\newcommand\integrand[3][1]{\ui{I_{#2,#3}}{#1}}
	In both cases, we find recurrences for the polynomial long division of the integrand where
	we first decrement $a$ and $b$ simultaneously by $1$ until $b$ becomes $0$.
	Then, we decrement $a$ by $3$ (for $c=1$) resp.\ by $4$ (for $c=2$).
	More precisely with 
	$\integrand[c]ab \ce \frac{x^a(1-x)^b}{\frac1c-x(1-x)}$, we have for $a\ge b\ge 0$:
	\begin{align*}
			\integrand[1]{a+1}{b+1} - \integrand[1]{a}{b}
		&\wwrel=
			-x^a(1-x)^b \text{ and}
	\\
			\integrand[1]{a+3}{b} + \integrand[1]ab
		&\wwrel=
			(x+1)x^a(1-x)^b .
	\end{align*}
	Together with the base cases
	\begin{align*}
		\integrand[1]00 &\wwrel= \frac 1{1-x(1-x)}, \\
		\integrand[1]10 &\wwrel= \frac x{1-x(1-x)}    \quad\text{and}\\
		\integrand[1]20 &\wwrel= 1 + \integrand[1]10 - \integrand[1]00 
	\end{align*} 
	we get a telescoping recurrence for the integrand $\integrand[1]ab$ 
	that allows to write the fraction as a polynomial in $x$ plus a 
	multiple of one of the base cases
	($\integrand[1]d0$ for $d=(a-b) \bmod 3$ to be precise).
	The advantage of this representation is that it can easily be integrated term-wise:
	The polynomials trivially and the base cases using suitable linear combinations of
	the anti-derivatives
	\begin{align*}
			\int \frac1{1-x(1-x)} \,dx 
		&\wwrel= 
			\frac2{\sqrt3} \arctan\biggl( \frac{2}{\sqrt3}x-\frac1{\sqrt3} \biggr)
	\\
			\int \frac{2x-1}{1-x(1-x)} \,dx 
		&\wwrel= 
			\ln\bigl( 1-x(1-x) \bigr)
	\end{align*}
	We thus obtain a recurrence for the integrals, which again directly telescopes.
	Simplifying the involved Beta function terms gives the
	``closed form'' shown in \wref{fig:geometric-beta-integrals}.
	(For the sums involved in this term, we could not find more succinct representations.)
	This finishes the proof of the first part ($c=1$).
	
	For $c=2$, the recurrence is only slightly different:
	\begin{align*}
			\integrand[2]{a+1}{b+1} - \tfrac12\integrand[2]{a}{b}
		&\wwrel=
			-x^a(1-x)^b \text{ and}
	\\
			\integrand[2]{a+4}{b} + \tfrac14\integrand[2]ab
		&\wwrel=
			(\tfrac12 + x(x+1))x^a(1-x)^b .
	\end{align*}
	with the base cases
	\begin{align*}
		\integrand[2]00 &\wwrel= \frac 1{\tfrac12-x(1-x)}, \\
		\integrand[2]10 &\wwrel= \frac x{\tfrac12-x(1-x)}    \\
		\integrand[2]20 &\wwrel= 1 + \integrand[2]10 - \tfrac12 \integrand[2]00  \quad\text{and}\\
		\integrand[2]30 &\wwrel= x+1 + \tfrac12\integrand[2]10 - \tfrac12\integrand[2]00  .
	\end{align*}
	Term-wise integration gives again a telescoping recurrence for the geometric beta 
	integral with $c=2$, where we use the anti-derivatives
	\begin{align*}
			\int \frac1{\frac12-x(1-x)} \,dx 
		&\wwrel= 
			-2\arctan( 1 - 2x )
	\\
			\int \frac{2x-1}{\frac12-x(1-x)} \,dx 
		&\wwrel= 
			\ln\bigl( 1 - 2x(1-x) \bigr)
	\end{align*}
	to integrate the base cases.
	The remaining steps are as for the $c=1$ case.
\end{proof}

%% file: branch-misses-computations.tex
\section{Detailed Computations of Expected Toll Functions}
\label{app:computations}

In this appendix, we give detailed computations of the expected
number of branch misses in the first partitioning steps of CQS and YQS, 
respectively,
which means in particular to derive the values for 
$a_{\mathrm{CQS}}$ and $a_{\mathrm{YQS}}$ given in \wref{thm:main-results-BM}.

The computations heavily use the lemmas about Dirichlet distributed
random variables from \wref{app:distributions}, and
the following notation is excessively used:
We write $\dirichletExpectation{f(\vect X)}{\vect\alpha}$
for the expectation of any term 
$f(\vect X)$ that depends on $\vect X$
where $\vect X \eqdist \dirichlet(\vect \alpha)$, \ie,
formally for $\vect\alpha \in \R_{>0}^d$
$$
		\dirichletExpectation{f(\vect X)}{\vect\alpha} 
	\wwrel= 
		\int_{\Delta_d} \mkern-10mu f(\vect x) 
			\frac{x_1^{\alpha_1-1}\cdots x_d^{\alpha_d-1}}{\BetaFun(\vect\alpha)}
		\mu(d\vect x)
	\,,
$$
where $\Delta_d$ is the standard $(d-1)$-dimensional simplex, see \wref{eq:def-delta-d}.
Here, $\vect X$ is to be understood as a formal parameter, \ie, 
a local, bound variable whose distribution potentially differs between two 
$\dirichletExpectation{\circ}{\circ}$ terms.
We can express the abbreviation terms $\geoDirichletExp[c]ab$ 
(see \wildtpageref[eq.]{eq:def-geo-dirichlet-exp}{\eqref})
concisely as Dirichlet expectation:
\begin{align}
\label{eq:geo-dirichlet-exp-via-dirichlet-expectation}
		\geoDirichletExp[c]ab
	&\wwrel{=}
		\dirichletExpectation{\frac{X_1 X_2}{\frac1c - X_1 X_2}}{a,b} \;.
\end{align}

Moreover, we abbreviate $\tp_i = t_i + 1$ and $\kp = k+1 = \tp_1+\cdots+\tp_s$, 
as those terms occur very often below, so that we need a more convenient
way to write them.

\todoin[disable]{%
	{\itshape Shall we retain the following `intuition'? 
	It is clear that it has something to do with BM, but it is no way near rigorous
	and not really helpful for the analysis \dots}
	
	\medskip\noindent
	\textbf{Remark:}
	There is some intuition behind the term
	$$
			\geoDirichletExp[2]{\alpha_1}{\alpha_2}
		\wwrel{=}
			\dirichletExpectation{\frac{2X_1 X_2}{1 - 2X_1 X_2}}{\alpha_1,\alpha_2} \;:
	$$
	It gives the expected odds for observing two \emph{different} outcomes in two independent repetitions
	of a Bernoulli trial with parameter $p=X_1$ against observing two times the \emph{same} outcome, 
	where the
	parameter $X_1$ is itself random (namely $\dirichlet(\alpha_1,\alpha_2)$ distributed).
	
	This coincides with the expected number of times the two Bernoulli trials show different outcomes
	before they first show the same outcome.
}

\subsection{Classic Quicksort}

As already discussed in the main text, for CQS the two comparison location behave 
symmetrically, which simplifies analysis.
More precisely, we have for $\steadystatemissrate{}$ the steady-state miss-rate function of the 
used branch prediction scheme
\begin{align*}
		\E[\toll [n] \branchmisses]
	&\wwrel=
		      \E[\ui{C_n}{c1} \steadystatemissrate{}(\btprob{c1}) ]
		\bin+ \E[\ui{C_n}{c2} \steadystatemissrate{}(\btprob{c2}) ]
\\&\wwrel=
		\bigl(
			  \E[D_1 \steadystatemissrate{}(D_1) ]
			+ \E[D_2 \steadystatemissrate{}(D_2) ]
		\bigr) n \bin+ \Oh(1)
\\&\wwrel{\relwithtext[r]{symmetry of $\steadystatemissrate{}$}=}
		\E[(D_1+D_2) \steadystatemissrate{}(D_1)] \,n \bin+ \Oh(1)
\\&\wwrel=
		\E[\steadystatemissrate{}(D_1)]\,n \bin+ \Oh(1)\;.
\end{align*}
The second equation uses the quantities given in \wref{tab:execution-freqs-and-branch-probs-cond-on-D}.
This gives a more analytical proof of 
\wildpageref[equation]{eq:toll-CQS-generic}{\eqref}.
With the abbreviation 
$$
		g_{x,y}
	\wwrel=
		\dirichletExpectation{\steadystatemissrate{}(X_1)}{x,y}
	\wwrel=
		\dirichletExpectation{\steadystatemissrate{}(X_2)}{x,y}
$$ 
we obtain $a_{\mathrm{CQS}} = g_{\tp_1,\tp_2}$ 
as claimed in \wref{thm:main-results-BM}\,---\,it only remains to show that, 
upon inserting the actual steady-state miss-rate functions
(collected in \wref{tab:steady-state-miss-rate-functions-in-q} for convenience),
$g_{x,y}$ has the following explicit terms:
\begin{enumerate}[label=(\roman*)]
	\item \mbox{\makeboxlike[l]{2-bit sc: }{1-bit:}}
	$
			g_{x,y}
		=
			2xy/\rf{(x+y)}2
	$,
	\item \mbox{\makeboxlike[l]{2-bit sc: }{2-bit sc:}}
	$
			g_{x,y}
		=
			\frac12 \geoDirichletExp[2]{x}{y}
	$,
	\item \mbox{\makeboxlike[l]{2-bit sc: }{2-bit fc:}}
	$
			g_{x,y}
		=
			\frac{2xy}{\rf{(x+y)}{2}}\geoDirichletExp[1]{x+1}{y+1}
			+\geoDirichletExp[1]{x}{y}
	$.
\end{enumerate}

\begin{table}
	\small
	\plaincenter{%
		\begin{tabular}{rc}
			\toprule
			     Scheme & Miss Rate ($q=p(1-p)$) \\
			\midrule
			  \bmonebit &          $2q$          \\[1ex]
			\bmtwobitsc &   $\dfrac{q}{1-2q}$   \\[2ex]
			\bmtwobitfc & $\dfrac{2q^2+q}{1-q}$  \\
			\bottomrule
		\end{tabular}
	}
	\caption{%
		The steady-state miss-rate functions for our 
		branch prediction schemes, given in terms of $q = p(1-p)$.
		This change of variable is possible since the functions are
		symmetric:
		$\protect\steadystatemissrate{\mathit{bps}}(p)=\protect\steadystatemissrate{\mathit{bps}}(1-p)$
		for all $p\in[0,1]$
	}
	\label{tab:steady-state-miss-rate-functions-in-q}
\end{table}

We compute
\begin{align*}
		\dirichletExpectation{\steadystatemissrate{\bmonebit} (X_1)}{x,y}
	&\wwrel=
		\dirichletExpectation{2 X_1 X_2}{x,y}
\\	&\wwrel{\eqwithref[r]{lem:dirichlet-powers-to-parameters}}
		\frac{2 x y}{\rf{(x+y)}2}
\end{align*}
for 1-bit prediction;
\begin{align*}
		\dirichletExpectation{\steadystatemissrate{\bmtwobitsc} (X_1)}{x,y}
	&\wwrel=
		\dirichletExpectation{\frac{X_1 X_2}{1-2X_1 X_2}}{x,y}
\\	&\wwrel=
		\frac12 \dirichletExpectation{\frac{X_1 X_2}{\frac12-X_1 X_2}}{x,y}
\\	&\wwrel{\eqwithref[r]{eq:geo-dirichlet-exp-via-dirichlet-expectation}} 
	\frac12 \geoDirichletExp[2]{x}{y}
\end{align*}
for the 2-bit sc predictor and finally with the 2-bit fc predictor
\begin{align*}
		\dirichletExpectation{\steadystatemissrate{\bmtwobitfc} (X_1)}{x,y}
	\mkern-50mu\\
	&\wwrel=
		\dirichletExpectation{\frac{2(X_1 X_2)^2 + X_1X_2}{1-X_1X_2}}{x,y}
\\	&\wwrel{\eqwithref[r]{lem:dirichlet-powers-to-parameters}}
		\frac{2 x y}{\rf{(x+y)}2} 
		\dirichletExpectation{\frac{X_1X_2}{1-X_1X_2}}{x+1,y+1}
\\* &\wwrel\ppe \qquad{}
		\bin+ \dirichletExpectation{\frac{X_1X_2}{1-X_1X_2}}{x,y}
\\	&\wwrel{\eqwithref[r]{eq:geo-dirichlet-exp-via-dirichlet-expectation}} 
		\frac{2 x y}{\rf{(x+y)}2} \geoDirichletExp[1]{x+1}{y+1}
		\bin{+}
		\geoDirichletExp[1]{x}{y}.
\end{align*}
Together with \wref{thm:leading-term-expectation-hennequin},
this proves the part of \wref{thm:main-results-BM} concerning CQS.

\begin{table}
	\small
	\plaincenter{%
	\begin{tabular}{lcc}
		\toprule
		Location $i$ & $\E[\ui {C_n}i\given \vect D] \big/ n$ &      $\btprob i$      \\
		\midrule
		$i=c1$       &                    $D_1$                    &         $D_1$         \\
		$i=c2$       &                    $D_2$                    &         $D_2$         \\
		$i=c$        &                     $1$                     &         $D_1$         \\
		\midrule
		$i=y1$       &                  $D_1+D_2$                  &       $D_2+D_3$       \\
		$i=y2$       &            $(D_1+D_2)(D_2+D_3)$             & $\frac{D_2}{D_2+D_3}$ \\
		$i=y3$       &                    $D_3$                    &       $D_1+D_2$       \\
		$i=y4$       &               $D_3(D_1+D_2)$                & $\frac{D_1}{D_1+D_2}$ \\
		\bottomrule
	\end{tabular}%
	}
	\caption{%
		Summary of the quantities used for the analysis of 
		branch misses in CQS ($i=c1,c2,c$) and YQS ($i=y1,y2,y3,y4$).
		For the execution frequency,
		the table lists only the leading term coefficients, \eg,
		$\E[\protect\ui{C_n}{c1}\given \vect D] = D_1 n + \Oh(1)$.
	}
	\label{tab:execution-freqs-and-branch-probs-cond-on-D}
\end{table}

\subsection{Yaroslavskiy's Quicksort}
In YQS, we have four instead of just two comparisons locations,
and here, no symmetries can be exploited to ease computations.
Nevertheless, essentially the same procedure as for CQS remains 
possible when using further properties of Dirichlet vectors.

We consider the four comparison locations of YQS separately 
to keep the presentation accessible.
For convenience, \wref{tab:execution-freqs-and-branch-probs-cond-on-D} 
collects the needed information about the comparison locations:
how often they are executed and with which probability the corresponding
branches are taken, both averaged over all choices for the ordinary elements,
but conditional on fixed $\vect D$, \ie, fixed pivot values.

\subsubsection[$C^{y1}$]{\boldmath $\ui C{y1}$.}
For the first comparison location, we get (using linearity of the
expectation and \wref{lem:dirichlet-powers-to-parameters})
\begin{align*}
		\E[\ui {C_n}{y1} \steadystatemissrate{}(\btprob{y1})] / n
	\mkern-100mu\\
	&\wwrel=
		\dirichletExpectation{(X_1+X_2) \cdot \steadystatemissrate{}(X_2+X_3)}{\vect\tp} 
		\bin+ \Oh(\tfrac1n)
\\	&\wwrel= 
		\phantom{{}\bin+{}}
		\frac{\tp_1}{\kp} \dirichletExpectation{\steadystatemissrate{}(X_2+X_3)}{\vect\tp+(1,0,0)}
\\* &\wwrel\ppe {}
		\bin+ \frac{\tp_2}{\kp} \dirichletExpectation{\steadystatemissrate{}(X_2+X_3)}{\vect\tp+(0,1,0)}
		\bin+ \Oh(\tfrac1n).
\end{align*}
Expectations of this form are easily dealt with using \textsl{aggregation} 
(\wref{lem:dirichlet-aggregation}):
\begin{align*}
		\dirichletExpectation{\steadystatemissrate{}(X_2+X_3)}{\vect\tp+(1,0,0)}
	&\wwrel=
		\dirichletExpectation{\steadystatemissrate{}(X_2)}{\tp_1+1,\tp_2+\tp_3}
\\	&\wwrel=
		g_{\tp_1+1,\tp_2+\tp_3},
\\[1ex]
		\dirichletExpectation{\steadystatemissrate{}(X_2+X_3)}{\vect\tp+(0,1,0)}
	&\wwrel=
		\dirichletExpectation{\steadystatemissrate{}(X_2)}{\tp_1,\tp_2+\tp_3+1}.
\\	&\wwrel=
		g_{\tp_1,\tp_2+\tp_3+1}.
\end{align*}
(Note that $\vect X$ is a three-dimensional vector on the left
hand sides and a two-dimensional one on the right hand sides.)

The total contribution of the first comparison location to the expected 
number of branch misses is then simply
\begin{multline*}
		\E[\ui {C_n}{y1} \steadystatemissrate{}(\btprob{y1})] 
	\wwrel=
	\\
		\Bigl(
			      \frac{\tp_1}{\kp} g_{\tp_1+1,\tp_2+\tp_3}
			\bin+ \frac{\tp_2}{\kp} g_{\tp_1,\tp_2+\tp_3+1}
		\Bigr) n \bin+ \Oh(1)
\end{multline*}

\subsubsection[$C^{y2}$]{\boldmath $\ui C{y2}$.}
For the second comparison location, the involved terms get a little messier. 
The leading term coefficient of 
$\E[\ui {C_n}{y2} \steadystatemissrate{}(\btprob{y2})]$ is
(cf.\ \wref{tab:execution-freqs-and-branch-probs-cond-on-D})
\begin{align*}
		\dirichletExpectation{(X_1+X_2)(X_2+X_3) \cdot
			\steadystatemissrate{}\biggl(\frac{X_2}{X_2+X_3}\biggr)}{\vect\tp}.
\end{align*}
After expanding $(X_1+X_2)(X_2+X_3)$ and splitting the 
expectation into summands,
we can use the \textsl{powers-to-parameters rule} 
(\wref{lem:dirichlet-powers-to-parameters})
to get four simpler terms, the first two of which are
\begin{multline*}
		\dirichletExpectation{X_1 X_2 \,
				\steadystatemissrate{}\bigl(\tfrac{X_2}{X_2+X_3}\bigr)}{\vect\tp}
	\wwrel=\\
		\tfrac{\tp_1\tp_2}{\rf\kp2}
			\dirichletExpectation{\steadystatemissrate{}\bigl(\tfrac{X_2}{X_2+X_3}\bigr)}
			{\vect\tp+(1,1,0)},
\end{multline*}
\vspace{-4ex}
\begin{multline*}
		\dirichletExpectation{X_2 X_2 \,
				\steadystatemissrate{}\bigl(\tfrac{X_2}{X_2+X_3}\bigr)}{\vect\tp}
	\wwrel=\\
		\tfrac{\rf{\tp_2}2}{\rf\kp2}
			\dirichletExpectation{\steadystatemissrate{}\bigl(\tfrac{X_2}{X_2+X_3}\bigr)}
			{\vect\tp+(0,2,0)},
\end{multline*}
and the remaining two are similar.
Here, the argument of the steady-state miss-rate function $\steadystatemissrate{}$
is a fraction, more precisely, it is the ratio between component $X_2$ and the
``subtotal'' $X_2+X_3$.
This means, the rest of the vector (namely $X_1$) is immaterial for the expectation
and we can \textsl{``zoom''} in using \wref{lem:dirichlet-zoom}:
\begin{align*}
		\dirichletExpectation{\steadystatemissrate{}\bigl(\tfrac{X_2}{X_2+X_3}\bigr)}
				{\vect\tp+(1,1,0)}
	&\wwrel=
		\dirichletExpectation{\steadystatemissrate{}(X_1)}
				{\tp_2+1,\tp_3}
\\	&\wwrel=
		g_{\tp_2+1,\tp_3},
\end{align*}
and similarly for the other terms.

Adding up the four summands, the overall contribution of $\ui C{y2}$ is then
\begin{multline*}
		\E[\ui {C_n}{y2} \steadystatemissrate{}(\btprob{y2})] 
	\wwrel=
	\\
			\biggl(
				      \frac{\tp_1\tp_2}{\rf\kp2} g_{\tp_2+1,\tp_3}
				\bin+ \frac{\tp_1\tp_3}{\rf\kp2} g_{\tp_2,\tp_3+1}
				\bin+ \frac{\rf{\tp_2}2}{\rf\kp2} g_{\tp_2+2,\tp_3}
\\				\bin+ \frac{\tp_2\tp_3}{\rf\kp2} g_{\tp_2+1,\tp_3+1}
			\biggr) n \bin+ \Oh(1).
\end{multline*}

\subsubsection[$C^{y3}$]{\boldmath $\ui C{y3}$.}

The third comparison location is very similar to $\ui C{y1}$
(in fact even a little simpler), so we only give the main steps:
\begin{align*}
		\E[\ui {C_n}{y3} \steadystatemissrate{}(\btprob{y3})] / n
	\mkern-50mu\\
	&\wwrel=
		\dirichletExpectation{X_3 \cdot
			\steadystatemissrate{}(X_1+X_2)}{\vect\tp} \bin+ \Oh(\tfrac1n)
\\	&\wwrel=
		\frac{\tp_3}{\kp} \dirichletExpectation{
					\steadystatemissrate{}(X_1)}{\tp_1+\tp_2,\tp_3+1}
					 \bin+ \Oh(\tfrac1n)
\\	&\wwrel=
		\frac{\tp_3}{\kp} g_{\tp_1+\tp_2,\tp_3+1}  \bin+ \Oh(\tfrac1n).
\end{align*}

\subsubsection[$C^{y4}$]{\boldmath $\ui C{y4}$.}

The last comparison location is similar to $\ui C{y2}$;
the main trick is again to use the \textsl{zooming lemma}.
We compute
\begin{align*}
		\E[\ui {C_n}{y4} \steadystatemissrate{}(\btprob{y4})] / n
	\mkern-130mu\\
	&\wwrel=
		\dirichletExpectation{X_3(X_1+X_2) \cdot
			\steadystatemissrate{}\biggl(\frac{X_1}{X_1+X_2}\biggr)}{\vect\tp}
			 \bin+ \Oh(\tfrac1n)
\\	&\wwrel= \phantom{{}\bin+{}}
		\frac{\tp_1\tp_3}{\rf\kp2} \dirichletExpectation{
			\steadystatemissrate{}\bigl(\tfrac{X_1}{X_1+X_2}\bigr)}{\vect\tp+(1,0,1)}
\\* &\wwrel\ppe {}
		\bin+\frac{\tp_2\tp_3}{\rf\kp2} \dirichletExpectation{
			\steadystatemissrate{}\bigl(\tfrac{X_1}{X_1+X_2}\bigr)}{\vect\tp+(0,1,1)}
		\bin+ \Oh(\tfrac1n)
\\	&\wwrel=
		      \frac{\tp_1\tp_3}{\rf\kp2} g_{\tp_1+1,\tp_2}
		\bin+ \frac{\tp_2\tp_3}{\rf\kp2} g_{\tp_1,\tp_2+1}
		 \bin+ \Oh(\tfrac1n).
\end{align*}

\bigskip\noindent
Adding up the contributions of all four comparison locations, we obtain 
the the value claimed for $a_{\mathrm{YQS}}$ in \wref{thm:main-results-BM}.